\documentclass[a4paper,10pt]{article}
\usepackage{amsmath,amssymb,amsthm,graphicx,color,enumerate,mathrsfs,dsfont,url,hyperref,comment,tikz,algorithm, algpseudocode}
\usepackage[normalem]{ulem}
\usepackage{ marvosym }

\definecolor{orange}{RGB}{230,159,0}
\definecolor{blue}{RGB}{0,114,178}
\definecolor{red}{RGB}{213,94,0}
\definecolor{trajgreen}{RGB}{0,158,115}

\tikzset{
  traj0/.style={orange, line width=1.0pt},
  traj1/.style={blue, line width=1.5pt, dash pattern={on 7pt off 3pt}},
  traj2/.style={red, line width=2.0pt, dash pattern={on 3pt off 2pt on 1pt off 2pt}},
  traj3/.style={trajgreen, line width=2.5pt, dash pattern={on 2pt off 2pt}},
  hyp0/.style={traj0, opacity=0.45},
  hyp1/.style={traj1, opacity=0.45},
  hyp2/.style={traj2, opacity=0.45},
  hyp3/.style={traj3, opacity=0.45},
  hyp2b/.style={traj2, opacity=0.2},
}

\renewcommand{\thefootnote}{\fnsymbol{footnote}}
\renewcommand{\title}[1]{\vspace{\fill}
\eject\addtolength{\baselineskip}{4pt}
{\bfseries\LARGE #1}\\[3mm]\addtolength{\baselineskip}{-4pt}}
\renewcommand{\author}[3]{\parbox[t]{75mm}
{\begin{center}{\scshape #1}\\[3mm] #2\\
 {\ttfamily #3} \end{center}}}

\newcommand\numberthis{\addtocounter{equation}{1}\tag{\theequation}}

\newtheorem{thm}{\bfseries Theorem}
\newtheorem{lem}[thm]{\bfseries Lemma} 
\newtheorem{remark}[thm]{\bfseries Remark} 
\newtheorem{prop}[thm]{\bfseries Proposition} 
\newtheorem{cor}[thm]{\bfseries Corollary}
\newtheorem{defn}[thm]{\bfseries Definition}
\newtheorem{cl}[thm]{\bfseries Fact}

\numberwithin{thm}{section}

\renewcommand{\P}{\mathbb{P}}

\newcommand{\e}{\mathrm{e}}
\newcommand{\R}{\mathbb{R}}
\newcommand{\E}{\mathbb{E}}
\renewcommand{\P}{\mathbb{P}}

\newcommand{\N}{\mathbb{N}}

\newcommand{\smfrac}[2]{\textstyle{\frac{#1}{#2}}}

\begin{document}

\begin{center}

\title{Resident fitness computation in linear time and \\ other algorithmic aspects of interacting trajectories\footnote{A shorter preliminary version~\cite{FNT25} of this paper appeared in the Proceedings of the 13th Hungarian--Japanese Symposium on Discrete Mathematics and Its Applications.}}
\author{Katalin Friedl\footnotemark[0]
}{
Department of Computer Science and Information Theory \\
Budapest University of Technology and Economics \\
Műegyetem rkp. 3., H-1111 Budapest, Hungary
}{
friedl@cs.bme.hu
}
\author{Viktória Nemkin\footnotemark[2]
}{
Department of Computer Science and Information Theory \\
Budapest University of Technology and Economics \\
Műegyetem rkp. 3., H-1111 Budapest, Hungary
}{
nemkin@cs.bme.hu
}

\author{
\underline{András Tóbiás}\footnotemark[4]
}{
Department of Computer Science and Information Theory \\
Budapest University of Technology and Economics \\
Műegyetem rkp. 3., H-1111 Budapest, Hungary; \\
HUN-REN Alfréd Rényi Institute of Mathematics \\
Reáltanoda utca 13--15, H-1053 Budapest, Hungary
}{
tobias@cs.bme.hu
}
\footnotetext[2]{This paper was supported by the Doctoral Excellence Fellowship Programme (DCEP), funded by the National Research, Development and Innovation Fund of the Ministry of Culture and Innovation and the Budapest University of Technology and Economics.}
\footnotetext[4]{This paper was supported by the János Bolyai Research Scholarship of the Hungarian Academy of Sciences. Project no.\ STARTING 149835 has been implemented with the support provided by the Ministry of Culture and Innovation of Hungary from the National Research, Development and Innovation Fund, financed under the STARTING\_24 funding scheme.}

\end{center}

\begin{quote}
{\bfseries Abstract:}
Systems of interacting trajectories were recently studied in~\cite{HGSTW24}.
Such a system of $[0,1]$-valued piecewise linear trajectories arises as a scaling limit of the system
of logarithmic subpopulation sizes in a population-genetic model (more precisely, a
Moran model) with mutation and selection. By definition, the resident fitness is initially 0
and afterwards it increases by the ultimate slope of each trajectory that reaches height 1.

We show that although the interaction of $n$ trajectories may yield $\Omega(n^2)$ slope changes in
total, the resident fitness function can be computed algorithmically in $O(n)$ time. Our
algorithm uses the so-called continued lines representation of the system of
interacting trajectories. In the special case of Poissonian interacting trajectories where the
birth times of the trajectories form a Poisson process and the initial slopes are random and
i.i.d., we provide a linear bound on the expected total number of slope changes.

\end{quote}

\begin{quote}
\end{quote}
\vspace{5mm}

\renewcommand{\thefootnote}{\arabic{footnote}}

\section{Introduction}\label{sec-intro}

The concept of a system of interacting trajectories has been recently introduced in~\cite{HGSTW24}. Such a system consists of continuous and piecewise linear trajectories in a 2D coordinate system, where the $x$ axis represents time and the $y$ axis represents a magnitude that corresponds to the size of subpopulations of a large population of constant total size on a logarithmic scale. The realization of the trajectories is in $[0,1]$ at each time. Each trajectory has a birth time and an initial slope. Until the birth time, the trajectory is constant zero, then it starts growing linearly with its initial slope. When a trajectory reaches height 1, it becomes constant 1, while the slope of all other trajectories currently at a positive height is reduced by the ultimate slope of this trajectory before reaching height 1. The slope of each trajectory is nonincreasing, and trajectories that reach height 0 again stay constant 0 forever. This interactive dynamics models the exponential growth/decay of mutant subpopulations in a certain stochastic population-genetic model, more precisely a Moran model, with mutation and selection, under a logarithmic scaling. Here, selection is strong and the mutation rate is in the so-called Gerrish--Lenski regime (dating back to~\cite{GL98}), see Section~\ref{rem-modsel}.

A special case is the \emph{system of Poissonian interacting trajectories} (PIT), where birth times are random, forming a Poisson process, and the positive initial slopes are random and i.i.d.\ (i.e.\ independent and identically distributed) and independent of the birth times. Under a suitable scaling of parameters, the family of logarithmic subpopulation sizes indeed converges to the PIT in distribution (in a suitable space with a suitable topology, see~\cite[Theorem 2.7]{HGSTW24}). In the PIT, the \emph{resident fitness} $F(t)$ is defined as the sum of the ultimate slope of trajectories reaching height 1 up to time $t$. If the distribution of initial slopes has a finite first moment, $F(t)/t$ converges almost surely to a finite deterministic number $\overline{v}$ called the \emph{speed of adaptation}, and it tends to infinity otherwise (cf.\ \cite[Theorem 2.8]{HGSTW24}). Moreover, if the initial slopes have a finite second moment, then the resident fitness process satisfies a functional central limit theorem (see~\cite[Theorem 2.12]{HGSTW24}). The proofs of the latter two assertions are based on a renewal argument.

A closed formula for $\overline{v}$ is only known in the case when the initial slopes are deterministic and constant. In that case, the speed was already computed in~\cite{BGPW19}; this result was put into the context of the PIT by~\cite[Proposition 2.10]{HGSTW24}. On the other hand, one can clearly see the computation of the realization of $F(t)/t$ for a large $t$ as a Monte Carlo simulation for the value of $\overline{v}$. In the case when the fitness advantages have finite variance, we also know from the functional central limit theorem that typical fluctuations of $F(t)/t$ around $\overline{v}$ are of order $1/\sqrt t$. Thus, in order to approximate $\overline{v}$, it is useful to find an efficient algorithm for the simulation of the resident fitness of (Poissonian) interacting trajectories.

The main result of the present paper, Theorem~\ref{theorem-algorithm}, states that there is a deterministic algorithm that computes the resident fitness for all times in the case of a system of $n$ interacting trajectories (with deterministic birth times and initial slopes, where the birth times have been sorted in advance) in $O(n)$ time. The existence of such a linear-time algorithm is \emph{a priori} not obvious because as we will show in Proposition~\ref{prop-cousinen} below, the interactions between the trajectories may yield $\Omega(n^2)$ slope changes in total. Our general conclusion is that computing the resident fitness at all times is in general substantially easier than determining the slopes of all piecewise linear trajectories in all pieces. Our algorithm is described in terms of the so-called \emph{continued lines representation}, a process introduced in~\cite{HGSTW24} that is in one-to-one correspondence with the PIT but consists of {half-lines} instead of piecewise linear broken lines. As Corollary~\ref{cor-k} asserts, in case all initial slopes are in $\{ 1,2,\ldots,k\}$ for some $k \in \mathbb{N}$, the piecewise constant slope functions of all trajectories can be computed in $O(kn)$ time based on our main algorithm.

Moreover, our Proposition~\ref{lemma-average} below states that in the case of the \emph{Poissonian} interacting trajectories, the expected number of slope changes up to time $t$ is $O(t)$ as $t \to \infty$. The proof of the latter assertion is based on a modified version of the aforementioned renewal argument of~\cite{HGSTW24} combined with classical large-deviation estimates. This proof clearly shows that although the worst-case interaction between $n$ trajectories may provide $\Omega(n^2)$ slope changes in total, scenarios where there are more than a sufficiently large constant times $t$ slope changes up to time $t$ are exponentially unlikely in the Poissonian case. Based on this proof and additional arguments from~\cite{HGSTW24}, we derive Corollary~\ref{cor-speedofkinking}, a strong law of large numbers for the number of slope changes up to time $t$ in the large-$t$ limit.

Let us note that apart from the PIT, there are many other recent examples of piecewise linear limiting processes arising in stochastic population-genetic~\cite{DM11} and population-dynamic~\cite{BPT23,BCS19,B24,CKS21,CMT21,EK23,EK21,P23} models under a logarithmic scaling. These processes originate from models with higher mutation rates than the one of~\cite{HGSTW24}. The limiting processes are typically deterministic and their dynamics is also a bit different from the one of the interacting trajectories, thanks to multiple mutations between two given types of individuals and also due to other biological phenomena depending on the model (e.g.\ horizontal gene transfer, asymmetric competition, dormancy etc.). It is possible that variants or extensions of our algorithm can be used for simulating these limiting processes efficiently; we defer such questions to future work.

The rest of this paper is organized as follows. In Section~\ref{sec-model} we recall the definition of a system of (Poissonian) interacting trajectories from~\cite{HGSTW24} and we provide a short summary of its biological motivation. In Section~\ref{sec-results} we present our results and in Section~\ref{sec-proofs} we prove them. Finally, in Section~\ref{sec-discussion} we briefly discuss some consequences of our assertions and proofs. More concretely, in Section~\ref{rem-fix} we explain some applications of our algorithm in determining ancestral relations and fixation events of trajectories. In Section~\ref{rem-givent} we discuss how to determine height functions of the trajectories and fixed-time values of the resident fitness efficiently using our algorithm and data structure, and in Section~\ref{sec-unordered} we comment on the case when the input of the algorithm is initially not sorted by the birth times of trajectories. Finally, in Section~\ref{sec-b} we explain that our algorithm extends to a certain variant of the system of interacting trajectories exhibiting jumps, and in Section~\ref{rem-modsel} we mention that this variant has been conjectured to arise as a scaling limit of a Moran model in a parameter regime where biological selection is weaker than in our original model, which is called the case of moderate selection.

\section{Setting}\label{sec-model}
The definition of a system of (Poissonian) interacting trajectories that we will use in the present paper is a somewhat simplified version of the definition appearing in~\cite[Sections 2.1 and 3.1]{HGSTW24} (see Remark~\ref{rem-defchanges} below for the differences between the two definitions). It is given as follows.

Writing $\N=\{1,2,\ldots\}$ and $\N_0=\{0,1,2,\ldots\}$, for $n \in \N_0$ we define $[n]=\{ 1,2,\ldots,n\}$, moreover we put $[\infty]=\N$. We denote the space of continuous and piecewise linear trajectories $h$ from $[0,\infty)$ to $[0,1]$ such that $h(t) \in [0,1]$ for all $t \geq 0$ by $\mathcal C_{\rm PL}$. Each $h\in \mathcal C_{\rm PL} $ has at time $t$ a {\em height} $h(t)$ and a (right) {\em slope}
\begin{equation}\label{defvh}
   v_{h}(t):= \lim\limits_{\delta \downarrow 0} \tfrac 1\delta(h(t+\delta)-h(t)).
\end{equation}

Assume that for some $n\in \N_0 \cup \{ \infty \}$ we are given a configuration of pairs $(t_i,a_i)$ consisting of a \emph{birth time} $t_i$ and an \emph{initial slope} $a_i$ \begin{equation}\label{defbeth}
   \beth=((t_i,a_i))_{1\le i < n+1}\in ([0,\infty) \times (0,\infty))^{[n]},
\end{equation} with $0 \leq t_1< t_2\cdots<t_n$ if $n < \infty$ and with $0 \leq t_1<t_2<\ldots$ and $\lim_{n\to\infty} t_n=\infty$ if $n=\infty$ (where we use the convention $\infty+1=\infty$). $\beth$ (beth) specifies that the trajectory $h_i$, $1\le i < n+1$, has height $0$ for $t \in [0,t_i]$ and right slope $a_i$ at its {\em birth time}~$t_i$.

\begin{defn}[System of interacting trajectories,~\cite{HGSTW24}]\label{defdyn}
For $\beth$ as in~\eqref{defbeth}, let
\begin{equation}\label{defsystH}
  \mathbb{H}
    = (h_i)_{0 \leq i < n+1}
    \in \big(\mathcal C_{\rm PL}\big)^{[n] \cup \{ 0\}}
\end{equation}
result from the following deterministic interactive dynamics on $(\mathcal C_{\rm PL})^{[n] \cup \{0\}}$, where we write $v_i:=v_{h_i}$ for $0\leq i <n+1$ in order to simplify notation:
    \begin{itemize}
        \item One initial trajectory $h_0 \in \mathcal C_{\rm PL}$ starts at height $h_0(0) = 1$ with slope $v_{0}(0)=0$. All other trajectories $h_i$, $1 \leq i < n+1$ start at height $h_i(0)=0$ with slope $v_i(0)=0$.
        \item Trajectories continue with constant slope until the next birth time is reached or one of the trajectories reaches either 1 from below or 0 from above.
        \item For $i \geq 1$, at the birth time $t_i$ the slope $v_i$ of trajectory $h_i$ jumps from $0$ to $a_i$.
        \item Whenever at some time $t$ at least one trajectory reaches height $1$ from below, the slopes of all trajectories whose height is in $(0,1]$ at time $t$ are simultaneously reduced by
        \[
          v^\ast
           := \max\limits\{ v_i(t-)\mid 0 \leq i < n+1 \mbox{ such that } h_i(t) = 1\},
        \]
        i.e.\ for all $0 \leq i < n+1$ with $h_i(t) > 0$ \begin{eqnarray}\label{newslope} v_i(t) := v_i(t-) - v^\ast.\end{eqnarray}
        \item Whenever a trajectory at some time $t$ reaches height $0$ from above, its slope is instantly set to $0$, and this trajectory then stays at height $0$ forever. \\ Let $e_i$ denote the time when this happens to the $i$-th trajectory during the sequential application of the heuristics (where we set $e_i=\infty$ if it never happens) where $0 \leq i < n+1$. For $t \geq 0$ we say that the $i$-th trajectory is \emph{alive} at time $t$ if $t_i \leq t < e_i$.
    \end{itemize}
\end{defn}
 We call this $\mathbb{H}$ {\em the system of interacting trajectories initiated by $\beth$}, and denote it by $\mathbb{H}(\beth)$.
An alternative, equivalent representation of the system of interacting trajectories is the so-called \emph{continued lines representation}, which plays an important role in the proof of Theorem~\ref{theorem-algorithm} and will be introduced in Section~\ref{rem-contlines} below. Instead of broken lines, this representation consists of nondecreasing half-lines starting at the birth times of the corresponding trajectories, which makes it more amenable for algorithmic investigations.

\begin{defn}[Resident change times, resident type, and resident fitness,~\cite{HGSTW24}]\label{defrch} Let $\mathbb{H}$ be as in~\eqref{defsystH}, following the dynamics specified in Definition~\ref{defdyn}.
\begin{itemize}
    \item The times at which one of the trajectories $h_i$, $i > 0$, reaches height~$1$ from below will be called the {\em resident change times}.
    \item For $t > 0$ we call
\begin{equation}\label{defrestype}
    \varrho(t) = \varrho^{(\beth)}(t)
     :=\operatorname*{arg\,max}_{i} \{ v_i(t-) \mid 0 \leq i < n+1, h_i(t)=1 \},
\end{equation}
the {\em resident type} at time $t$, and we put $\varrho(0)=\varrho^{(\beth)}(0)=0$.
\item With the definition $f(0)=0$, we define the {\em resident fitness} $f(t)$, $t\ge 0$, by decreeing that $f$ at any resident change time~$r$ has an upward jump with
\begin{equation}\label{finc}
    f(r)-f(r-) = \max\{v_i(r-)\mid 0 \leq i < n+1,\, h_i(r) =1\}
    \end{equation}
\end{itemize}
and remains constant between any two subsequent resident change times.
\end{defn}

\begin{figure}
\centering
\scalebox{0.85}{
\begin{tikzpicture}
  \draw[scale=3.7, black, ->] (0.9, 0) -- (3.9, 0) node[right] {time};
  \draw[scale=3.7, black, ->] (0.9, -0.05) -- (0.9, 1.08) node[above] {height};
  \draw[scale=3.7,black,thick] (0.87,1) node[left] {$1$};
    \draw[scale=3.7,black,thick] (0.87,0) node[left] {$0$};
  \draw[scale=3.7, blue,thick] (1,-0.05) node[below] {$ t_1=1$};
  \draw[scale=3.7, blue,thick] (1.45,0.5) node[above] {$a_1$};
  \draw[scale=3.7, orange,thick] (2.45,0.6) node[above] {$-a_1$};
  \draw[scale=3.7, red,thick] (1.7,-0.05) node[below] {$ t_2=1.8$};
   \draw[scale=3.7, red,thick] (1.8,0.1) node[above] {$a_2$};
    \draw[scale=3.7, red,thick] (2.24,0.22) node[above] {$a_2-a_1$};
       \draw[scale=3.7, trajgreen, line width=1.0pt] (2.34,0.1) node[above] {$a_3$};
        \draw[scale=3.7, blue,thick] (3.28,0.75) node[above] {$-a_3$};
           \draw[scale=3.7, red,thick] (3.88,0.25) node[above] {$a_2-a_1-a_3$};
  \draw[scale=3.7, blue,thick] (2.05,-0.03) node[below] {$ t_1+\smfrac{1}{a_1}$};
    \draw[scale=3.7, trajgreen, line width=1.0pt] (2.41,-0.05) node[below] {$ t_3=2.35$};
      \draw[scale=3.7, trajgreen, line width=1.0pt] (3.1,-0.03) node[below] {$t_3+\smfrac{1}{a_3}$};
      \draw[scale=3.7, red,thick] (3.82,-0.03) node[below] {$t'$};
      \draw[scale=3.7, domain=0.9:2, smooth, variable=\x, traj0] plot ({\x}, {1});
  \draw[scale=3.7, domain=1:2, smooth, variable=\x, traj1] plot ({\x}, {\x-1});
  \draw[scale=3.7, orange,thick] (1.55,1.15) node[below] {$0$};
  \draw[scale=3.7, blue,thick] (2.55,1.15) node[below] {$a_1=1$};
  \draw[scale=3.7, trajgreen, line width=1.0pt] (3.55,1.15) node[below] {$a_1+a_3=2.5$};
  \draw[scale=3.7, domain=1.8:2, smooth, variable=\x, traj2] plot ({\x}, {1.5*(\x-1.8)});
  \draw[scale=3.7, domain=2:3, smooth, variable=\x, traj0] plot ({\x}, {-1*(\x-2)+1});
   \draw[scale=3.7, domain=2:3.016667, smooth, variable=\x, traj2] plot ({\x}, {0.5*(\x-2)+0.3});
   \draw[scale=3.7, domain=2:3.016667, smooth, variable=\x, traj1] plot ({\x}, {1});
   \draw[scale=3.7, domain=2.35:3.016667, smooth, variable=\x, traj3] plot ({\x}, {1.5*(\x-2.35)});
  \draw[scale=3.7, domain=3.016667:3.68, smooth, variable=\x, traj1] plot ({\x}, {-1.5*(\x-3.016667)+1});
  \draw[scale=3.7, domain=3.01667:3.9, smooth, variable=\x, traj3] plot ({\x}, {1});
  \draw[scale=3.7, domain=3.01667:3.825, smooth, variable=\x, traj2] plot ({\x}, {-1*(\x-3.016667)+0.808335});
\end{tikzpicture}
}

\scalebox{0.85}{
\begin{tikzpicture}
  \draw[scale=3.7, black, ->] (0.9, 0) -- (3.95, 0) node[right] {time};
  \draw[scale=3.7, black, ->] (0.9, -0.0125) -- (0.9, 1.075) node[above] {height};
  \draw[scale=3.7,black,thick] (0.87,0.25) node[left] {$1$};
    \draw[scale=3.7,black,thick] (0.87,0) node[left] {$0$};
  \draw[scale=3.7, blue,thick] (1,-0.0125) node[below] {$ t_1=1$};
  \draw[scale=3.7, blue,thick] (2.25,0.35) node[above] {$a_1=1$};
  \draw[scale=3.7, red,thick] (1.7,-0.0125) node[below] {$ t_2=1.8$};
   \draw[scale=3.7, red,thick] (1.8,0.06125) node[above] {$a_2$};
  \draw[scale=3.7, blue,thick] (2.05,-0.00725) node[below] {$ t_1+\smfrac{1}{a_1}$};
    \draw[scale=3.7, trajgreen, line width=1.0pt] (2.41,-0.0125) node[below] {$ t_3=2.35$};
      \draw[scale=3.7, trajgreen, line width=1.0pt] (3.1,-0.0075) node[below] {$t_3+\smfrac{1}{a_3}$};
      \draw[scale=3.7, domain=0.9:3.95, smooth, variable=\x, traj0] plot ({\x}, {0.25});
  \draw[scale=3.7, domain=1:3.95, smooth, variable=\x, traj1] plot ({\x}, {0.25*(\x-1)});
  \draw[scale=3.7, orange,thick] (1.5,0.375) node[below] {$0$};
  \draw[scale=3.7, trajgreen, line width=1.0pt] (2.95,0.75) node[below] {$a_1+a_3=2.5$};
  \draw[scale=3.7, domain=1.8:3.95, smooth, variable=\x, traj2] plot ({\x}, {0.375*(\x-1.8)});
   \draw[scale=3.7, domain=2.35:3.95, smooth, variable=\x, traj3] plot ({\x}, {0.625*(\x-2.35)+0.0875});
\end{tikzpicture}
}

\caption{Top: This is an example of a system of interacting trajectories with $n=3$, $(t_1,t_2,t_3)=(1,1.8,2.35)$ and $(a_1,a_2,a_3)=(1,1.5,1.5)$. We see that except the yellow (smooth) $0$-th trajectory born at height 1 at time $0$ with initial slope $0$, each trajectory $h_i$ is born at time $t_i$ with initial slope $a_i$ and that at each resident change, the slope of each trajectory drops by the ultimate slope of the new resident before the resident change. In particular, the orange (dash-dotted) trajectory $h_2$ never becomes resident; at time $t'$ it becomes extinct and the green (dotted) trajectory remains the only trajectory at positive height. The numbers between heights $0$ and $1$ describe the piecewise values of the (right) slopes $v_i(t)$. The numbers above height 1 correspond to the piecewise constant values of the resident fitness $t \mapsto f(t)$ in the given time intervals. Until the blue (dashed) trajectory $h_1$ reaches height $1$ at time $t_1+\frac{1}{a_1}$, the resident fitness equals its initial value $0$. Then it takes two consecutive jumps whose sizes equal the ultimate slopes of the two new residents, respectively: First at time $t_1+\frac{1}{a_1}$ when the blue (dashed) trajectory $h_1$ becomes resident, it increases by $1$ (and $\varrho(t)$ changes from $0$ to $1$), and then at time $t_3+\frac{1}{a_3}$ when the green (dotted) trajectory becomes resident, it increases by $1.5$ (and $\varrho(t)$ changes to $2.5$). \\
Bottom: The corresponding continued lines representation (see Section~\ref{rem-contlines}). The slope of each half-line is presented next to the half-line with the same colour and line style.}\label{fig-kick}
\end{figure}

See the top image of Figure~\ref{fig-kick} for an illustration of Definitions~\ref{defdyn} and~\ref{defrch}.
The following remark is an adaptation of~\cite[Equations (3.5) and (3.6)]{HGSTW24} to our situation.
\begin{remark}[\cite{HGSTW24}]\label{remresi}
\begin{enumerate}
\item The resident fitness $f$ obeys
      \begin{equation}\label{falternative}
        f(t)
         = a_{\varrho(t)} + f(t_{\varrho(t)}), \quad t \ge 0
      \end{equation}
      with $a_0:=0$.
\item
  If $t\ge 0$ is not a resident change time, then $\varrho(t)$ is the {\em unique} $i > 0$ for which $h_i(t) =1$. If $t$ is a resident change time and $h_i$ is the only trajectory that reaches height 1 from below at time $t$, then $\varrho(t)=i$. More generally, $\varrho(t)$ in~\eqref{defrestype} is the fittest of all types that are at height $1$ at time~$t$.

\end{enumerate}
\end{remark}

An important special case is when $\beth$ is random, $n=\infty$, and
\begin{itemize}
    \item the role of $(t_i)_{i \in \N}$ is played by the sequence of arrival times $(T_i)_{i\in \N}$ of a Poisson process $(T_i)_{i\in \N}$ with intensity $\lambda>0$ (i.e.\ putting $T_0=0$, $T_i-T_{i-1}$, $i=1,2,\ldots$ are i.i.d.\ exponentially distributed random variables with parameter $\lambda$),
    \item whereas the role of $(a_i)_{i\in \N}$ is played by a sequence $(A_i)_{i\in \N}$ of i.i.d.\ strictly positive random variables that is independent of the Poisson process $(T_i)_{i\in \N}$, where the common distribution of each $A_i$ on $(0,\infty)$ is denoted by $\gamma$ (i.e.\ $\gamma$ is a Borel probability measure on $(0,\infty)$ such that for any $x\in \mathbb{R}$, $\gamma((-\infty,x])=\P(A_1 \leq x)$).
    \end{itemize}
    In this case, we call $\mathbb{H}(\beth)$ the \emph{system of Poissonian interacting trajectories} (or \emph{PIT} for short) with parameters $(\lambda,\gamma)$, or briefly as the $\mathrm{PIT}(\lambda,\gamma)$. This system arises naturally as the large population scaling limit of Moran models with recurrent beneficial mutations, see~\cite[Theorem 2.7]{HGSTW24}. The corresponding trajectories and their (right) slopes will also be denoted by capital letters $H_i$ resp.\ $V_i$ (note that they are deterministic functions of the random input $((T_j,A_j))_{j\in\mathbb{N}}$). The resident fitness $F$ that we mentioned in the introduction corresponds to $f$ from Definition~\ref{defrch} in the case of the PIT. The resident type at time $t \geq 0$ in this context will still be denoted by $\varrho(t)$. We will generally exclude initial conditions that are degenerate according to the following definition.
\begin{defn}
Let $\beth$ be as in~\eqref{defbeth}. We say that $\beth$ is \emph{degenerate} if the corresponding $(h_i)_{i\in [n] \cup \{ 0 \}}$ defined according to Definition~\ref{defdyn} satisfies at least one of the following conditions:
\begin{itemize}
    \item There exists $t >0$ and three pairwise distinct indices $i,j,k \in [n] \cup \{0\}$ such that $h_i(t)=h_j(t)=h_k(t)=1$.
    \item There exists $j \in [n] \cup \{ 0 \}$ such that the extinction time of $h_j$, i.e.\ $\inf \{ t > t_j \mid h_j(t)=0\}$, is a resident change time or equals $t_i$ for some $i \in [n]$.
    \item There exists $i \in [n]$ such that $t_i$ is a resident change time.
\end{itemize}
\end{defn}
Elementary properties of the Poisson process imply the following fact (see also the proof of~\cite[Lemma 2.3]{HGSTW24}).
\begin{cl}
Let $\beth=((T_i,A_i))_{i\in \N}$ as above. Then we have that
$ \P ( \beth \text{ is nondegenerate} ) = 1. $
\end{cl}

\begin{remark}\label{rem-defchanges}
Apart from often considering only finitely many trajectories instead of countably many, the main changes in the above definitions compared to their analogues in~\cite{HGSTW24} are the following:
\begin{itemize}
    \item We ignore general initial conditions with multiple trajectories at positive heights at time 0, used in the proof of the
large population limit result. (This way, our $\mathbb{H}(\beth)$ corresponds to $\mathbb{H}(\aleph,\beth)$ in~\cite[Section 3.1]{HGSTW24} with the special choice $\aleph=(1,0)$.) For simplicity we allow $t_1=0$ (and not only $t_1 >0$).
    \item~\cite{HGSTW24} uses a parametrisation for the PIT which
is different from the one used in the present paper. Specifically, the
random input for the PIT in~\cite{HGSTW24} is a Poisson process with intensity
measure $\lambda^* \mathrm{d} t \cdot \gamma^*(\mathrm{d} a)$ (with $(\lambda^*, \gamma^*)$ specified in~\cite[Remark 2.2]{HGSTW24}), while
the input of the PIT in the present paper is a Poisson point process
with intensity measure $\lambda \mathrm{d} t \cdot \gamma(\mathrm{d} a)$. The reason is that the original definition of the PIT includes trajectories whose slope and height stays zero forever. These describe mutant families going extinct rapidly by chance shortly after the mutation. This possibility is biologically relevant but algorithmically uninteresting, and in the present paper we ignore it in order to ease notation.
\end{itemize}
\end{remark}

\subsection{The continued lines representation}\label{rem-contlines}
An alternative, equivalent representation of the PIT, called the \emph{continued lines representation}, was provided in~\cite[Section 5.3]{HGSTW24b}. This representation can be generated in a recursive way without any calculations involving already decaying height functions, which turns out to be useful for algorithmic investigations. We describe this representation for a general system of interacting trajectories as follows.

Given a non-degenerate input $\beth=(t_i, a_i)_{0 \leq i < n+1}$, construct recursively a sequence of straight {half-lines} $\ell_0, \ell_1, \dots$, where
\begin{itemize}
    \item $\ell_0$ starts from the point $(0, 1)$ (towards the right) with slope 0,
    \item $\ell_1$ starts from the point $(t_1, 0)$ with slope $a_1$,
    \item \dots and with $\ell^r_k :=$ the pointwise maximum of $\ell_0, \ell_1 \dots \ell_k$,
    \item $\ell_{k+1}$ starts from the point $(t_{k+1},\ell^r_k(t_{k+1})-1)$ with slope $(\ell^r_k)'(t_{k+1}) + a_{k+1}$,
\end{itemize}

etc. In this way, the sequence $\ell^r_{n}$, $n=0, 1, 2, \dots$ increases pointwise to a piecewise linear, convex function $\ell^r$ (the upper envelope of the sequence $(\ell_k)$). It follows from the construction that each of the linear pieces of $\ell^{r}$ are identical to a piece of $\ell_i$ for some (uniquely determined) $i$. See the left image of Figure~\ref{fig-contlines} for an illustration.

There is a bijection between the system of interacting trajectories and the continued lines representation with the same input. This is stated in the following lemma, whose proof is elementary and therefore omitted. See Figure~\ref{fig-kick} for an illustration.
\begin{lem}\label{lem-contlines}
For a non-degenerate input $\beth=(t_i, a_i)_{0 \leq i < n+1}$, with the above notation, trajectory $h_i$, $0 \leq i < n+1$, of the system $\mathbb{H}(\beth)$ of interacting trajectories is given as $h_i(s) = \ell_i(s)-\ell^r(s)+1$ for times $s$
between $t_i$ and $\min\{t > t_i : \ell_i(t) - \ell^r(t) + 1 = 0\}$, and as $h_i(s) = 0$ for
$s$ outside of this time interval.
This identifies the resident change times as the times where the slope of $\ell^r$ changes, as well as the resident fitness between two subsequent such times as $f(t) = (\ell^r)'(t)$.

Conversely, given $\mathbb{H}(\beth)$ with $f$ defined as in Definition~\ref{defrch}, {half-line} $\ell_i$ of the corresponding continued representation can be expressed as $\ell_i(t) = \int_0^{t_i} f(s) \mathrm{d} s + (f(t_i)+a_i)(t-t_i)$ for $t \geq t_i$.

In other words, {half-line} $\ell_i$ starts from height $\ell^r(t_i) = \int_0^{t_i} f(s) \mathrm{d} s$ with slope $f(t_i)+a_i$.
\end{lem}

In this case, for $k \geq 0$, if $j \in \{0,1,\ldots,k\}$ satisfies $\ell^r(t_{k+1})=\ell_j(t_{k+1})$, we say that $j$ is the \emph{parent} of $k+1$ (note that there is a unique such $j$ that is well-defined under the assumption of nondegeneracy) and $k+1$ is a \emph{child} of $j$.

\section{Results}\label{sec-results}

For a sequence $(\beth^{(n)})_{n\in\N}$ of nondegenerate initial data $\beth^{(n)}=((t_i^{(n)},a_i^{(n)}))_{1 \leq i < n+1}$, writing $\mathbb{H}(\beth^{(n)}) = (h_i^{(n)})_{0 \leq i < n+1}$, we denote by
\[ \begin{aligned} k_n(\beth^{(n)})&  = \sum_{t \geq 0 \colon t \text{ resident change time of } \mathbb{H}(\beth^{(n)})} \sum_{i=0}^n  \mathds 1_{\{ v_i^{(n)}(t-) \neq v_i^{(n)}(t+)\} \cap \{ h_i^{(n)}(t)>0 \}} \\ & = \sum_{ t \geq 0 \colon t \text{ resident change time of } \mathbb{H}(\beth^{(n)})} \big| \{ i \in \{ 0,1,\ldots,n\} \colon h_i \text{ is alive at time } t \} \big| \end{aligned}  \numberthis\label{kn} \]
the total number of \emph{kinks} (slope changes) due to resident changes in the system of interacting trajectories $\mathbb{H}(\beth^{(n)})$.
See, for example, the top image of Figure~\ref{fig-kick}, at time $t=t_3+\frac{1}{a_3}$, all trajectories alive (dotted green, dashed blue and dash-dotted orange) change slope. From this definition, we have the following result.
\begin{prop}\label{prop-cousinen}
There exists a sequence $(\beth^{(n)})_{n\in\mathbb{N}}$ as above such that $k_n(\beth^{(n)})=\Omega(n^2)$.
\end{prop}

The proof of Proposition~\ref{prop-cousinen} can be found in Section~\ref{sec-cousinen}.

\begin{remark}
Apart from the slope changes corresponding to the definition of $k_n$ in~\eqref{kn}, at each time $t_i$, $i=1,\ldots,n$, there is an additional slope change where the slope of the $i$-th trajectory switches from $0$ to $a_i$. Moreover, at most $n$ trajectories eventually hit 0 again and then their slope drops to 0. Altogether, these are at most $2n$ additional slope changes, and precisely $2n$ unless there are multiple indices $j$ with the same value of $f(t_j) + a_j$ as the ultimate resident. (In the case of the PIT, the latter is only possible if $\P(A_1=c)>0$ holds for some $c>0$.) Hence, given Proposition~\ref{prop-cousinen}, ignoring these $O(n)$ additional slope changes will not yield substantial changes in the worst-case estimates of the number of kinks.
\end{remark}

Given Proposition~\ref{prop-cousinen}, computing the slopes of all trajectories in a system of $n$ interacting trajectories clearly takes $\Omega(n^2)$ time in the worst case: Even the size of any data structure representing the trajectories requires $\Omega(n^2)$ space in the worst case. One might think that even computing only $t \mapsto \varrho^{(\beth)}(t)$, the current resident type for all times takes a time that is quadratic in the size of $\beth$. However, the following theorem shows that this important function of the system of interacting trajectories can be found much more efficiently. Here, by computing $t \mapsto \varrho^{(\beth)}(t)$, a piecewise constant function, we mean producing an ordered array of its breakpoints together with the associated values. This representation allows evaluating the function at any $t$ via binary search in $O(\log n)$. We assume $\beth$ is ordered, i.e. the corresponding $t_i$'s are sorted in the input. See Section~\ref{sec-unordered} for a discussion on the unordered case.

\begin{thm}\label{theorem-algorithm}
The algorithmic complexity of finding $t \mapsto \varrho^{(\beth)}(t)$, $t \geq 0$, for any finite, nondegenerate, ordered $\beth$ of size $n$ is $\Theta(n)$.
\end{thm}

The proof of this theorem will be provided in Section~\ref{sec-algorithm}. The proof of the upper bound is constructive: we present an algorithm that computes the array describing the function in $O(n)$ time. The key idea is to use the continued lines representation of the system, which makes it possible to provide such an algorithm using a simple data structure,
despite the possibly large number of kinks.
As a byproduct, the algorithm also computes the {half-lines} in the continued lines representation of the system. From these outputs, the resident fitness function $t \mapsto f(t)$ is also easily computable, since it corresponds to the slopes of the trajectories of the residents, cf.\ Lemma~\ref{lem-contlines}.

Given this output, in Section~\ref{rem-fix} we explain how to determine the set of trajectories that reach fixation and the ancestral relations between these trajectories (the notion of fixation will also be explained there), 
and in Section~\ref{rem-givent} below, we summarize how to efficiently compute $h_i(t)$: first for a fixed $t\geq 0$ and fixed $i$, then for all $t\geq0$ with a fixed $i$, and finally for all $i$ and $t\geq0$. In Section~\ref{sec-unordered} we discuss the case of an unordered $\beth$ and
 finally, in Section~\ref{sec-b} we explain that our algorithm extends to a variant of the PIT (and the continued lines representation) which, as we will recall in Section~\ref{rem-modsel}, was conjectured in~\cite{HGSTW24} to arise as a scaling limit of the system of logarithmic frequencies of a Moran model in a different selection regime.

In the special case when we additionally know that each $a_i$ is contained in the set $[k]$ for some $k \geq 1$, the algorithm appearing in the proof of Theorem~\ref{theorem-algorithm} can easily be modified so that it also computes the piecewise slopes of all interacting trajectories (not only the resident fitness or the slopes of the {half-lines} of the continued lines representation) in $O(kn)$ time. This fact is reflected by the following corollary.

\begin{cor}\label{cor-k}
Let $k\in\mathbb{N}$. For any finite, nondegenerate, ordered $\beth$ of size $n$ such that each initial slope $a_i$ lies in $ [k]$, the family of piecewise constant slopes $(t \mapsto v_i(t))_{i\in [n]}$ can be computed via a deterministic algorithm of runtime $O(kn)$.
\end{cor}
The proof of Corollary~\ref{cor-k} can be found in Section~\ref{sec-k}.

Our next proposition shows that if one considers the PIT where the birth times of trajectories form a Poisson process, then even the expected number of kinks (due to resident changes) will only grow linearly in time. Here we use a notion of kinks~\eqref{Kt} that is slightly different from the one in~\eqref{kn} above and more convenient to work with; see Remark~\ref{rem-kinks} below for a related discussion.

Let $\lambda,\gamma$ as in Section~\ref{sec-model} and let us consider the $\mathrm{PIT}(\lambda,\gamma)$. For $t\geq 0$ we write $N_t = | \{ i \geq 1 \colon T_i \leq t \}|$ where $T_i$ is the birth time of the $i$-th trajectory (with the aforementioned convention that $T_0=0$). 
Define \emph{the total number of kinks due to resident changes up to $t$} as
\[ K_t =\sum_{0<s \leq t \colon s\text{ resident change time}} \ \sum_{i \in \N_0} \mathds 1_{\{ H_i'(t-) \neq H_i'(t+)\} \cap \{ H_i(t) >0 \}} . \numberthis\label{Kt} \]

\begin{prop}\label{lemma-average}
We have $\E(K_t)=O(t)$.
\end{prop}
The proof of Proposition~\ref{lemma-average} will be carried out in Section~\ref{sec-average}. This proof is based on an extension of the aforementioned renewal argument of~\cite{HGSTW24}, combined with classical large-deviation inequalities. The main message of this proof is that the clustering of trajectories that leads to many kinks (like in the construction in the proof of Proposition~\ref{prop-cousinen}) is exponentially unlikely in the number of participating trajectories if these trajectories follow each other directly in the PIT according to their birth times.

\begin{remark}[Relation between Propositions~\ref{prop-cousinen} and~\ref{lemma-average}]\label{rem-kinks}
$K_t$ is the number of resident changes up to time $t$ counted
with multiplicity, i.e. with the number of trajectories that change their
slopes due to this resident change. Note that putting $\beth^{(n)}=((T_i,A_i))_{i \in [n]}$, the total number of kinks at resident change times due to the interactions of the first $N_t$ trajectories is in general an upper bound for $K_t$, i.e.\
\[ K_t \leq k_{N_t}(\beth^{(N_t)}), \]
the difference $k_{N_t}(\beth^{(N_t)})-K_t$ being the number of those kinks in $\mathbb{H}(\beth^{(N_t)})$ that happen after time $t$. However, it will be apparent from the proof of Proposition~\ref{lemma-average} that following time 0, after a stopping time with finite expectation, a state will be reached where $H_i(t)=0$ for all but one $i$. Hence, by the strong Markov property, this state will be visited infinitely often, with i.i.d.\ waiting times of finite expectation inbetween. Since a trajectory can only have kinks before its height ultimately reaches $0$, it follows that the difference between $K_t$ and $k_{N_t}(\beth^{(N_t)})$ is of finite order. This way, Proposition~\ref{lemma-average} can indeed be seen as complementary to Proposition~\ref{prop-cousinen} in the Poissonian case.
\end{remark}

Based on the proof of Proposition~\ref{lemma-average} and the one of~\cite[Theorem 2.8]{HGSTW24}, in Section~\ref{sec-average} we will also derive the following corollary, which can be seen as a strong law of large numbers for the total number of kinks due to resident changes.

\begin{cor}\label{cor-speedofkinking}
Let $E_0=0 < E_1 < E_2 < \ldots$ be the \emph{renewal extinction times}, defined via
\[ E_i = \inf \{ t \geq E_{i-1} \mid \varrho(t) \neq \varrho(E_{i-1}) \text{ and only the resident type at time $t$ has a positive height} \} \numberthis\label{Ei} \]
for $i \geq 1$.
We have
\[ \lim_{t\to\infty} \frac{K_t}{t} = \frac{\E(K_{E_1})}{\E(E_1)}, \numberthis\label{speedofkinking} \]
almost surely. The limit on the right-hand side is in $(0,\infty)$ for any choice of $\gamma$.
\end{cor}
\begin{remark}\label{rem-analogy}
Corollary~\ref{cor-speedofkinking} is formally analogous to~\cite[Theorem 2.8]{HGSTW24} on the speed of adaptation: There, under the additional assumption that $\E(A_1)<\infty$, using the \emph{solitary resident change times} $L_0=0<L_1<L_2<\ldots$ defined as
\[ L_i = \inf \{ t \geq L_{i-1} \mid ~ t \text{ is a resident change time and $V_j(t) \leq 0$ for all $j \geq 1$} \}, \qquad i \geq 1, \]
it was shown that the resident fitness $t \mapsto F(t)$ satisfies
\[ \lim_{t\to\infty} \frac{F(t)}{t} = \frac{\E(F(L_1))}{\E(L_1)} \in (0,\infty), \numberthis\label{speedofkinkingformula} \]
almost surely. We will explain this analogy in more detail and exploit it in the proof of Corollary~\ref{cor-speedofkinking} below. Note however that a substantial difference between the two results is that if $\E(A_1)=\infty$, then $\lim_{t\to\infty} F(t)/t=\infty$ a.s., while in Corollary~\ref{cor-speedofkinking} the limit is always finite.
\end{remark}
\begin{remark}
As an example, assuming that the system of interacting trajectories depicted in Figure~\ref{fig-kick} is a realization of the $0$-th, first, second, and third trajectory of the PIT, the realization of the first solitary resident change time $L_1$ corresponds to the time $t_3+\frac{1}{a_3}$ when the third (dotted green) trajectory becomes resident, given that the realization $t_4$ of the next birth time is larger than this time. Moreover, the realization of the first renewal extinction time $E_1$ corresponds to the time $t'$ when the second (dash-dotted orange) trajectory $h_2$ becomes extinct (i.e.\ hits zero from above), given that $t_4$ is even larger than $t'$. In general, each renewal extinction time is preceded by a solitary resident change time with no resident changes in between, but between two consecutive solitary resident change times there need not be a renewal extinction time.
\end{remark}

\begin{remark}
Let us consider the special case when $A_1$ (and each $A_i$, $i \geq 1$) is constant and equal to $c$. Then, in the corresponding
Moran model each mutation has the same selective advantage compared to the current resident population (cf.\ Section~\ref{rem-modsel}).

This case seems to be the most amenable to explicit computations in general. For example, we can easily show that the limit in~\eqref{speedofkinking} is $2\lambda$ in this case. To see this, note that if all $A_i$'s are equal to $c$, then any trajectory born after time $0$ will start with slope $c$ at height 0, reach slope $0$ at the first resident change following its birth and reach slope $-c$ at the next resident change. This slope will be preserved until the trajectory reaches height 0 (see also Lemma~\ref{lem-k}). This accounts for 2 slope changes due to resident changes per trajectory, apart from the $0$-th trajectory, which starts from height 1 with slope 0 and suffers only one slope change due to a resident change. Thus, at time $t$, $K_t$ equals $2N_t+1$ minus the current number of trajectories at a positive height with zero slope, minus twice the current number of trajectories at a positive height with positive slope (necessarily equal to $c$). The trajectories at a positive height with zero slope are precisely the current resident and its siblings (i.e.\ the trajectories that were born when the parent of the current resident was resident), and hence their number is not greater than a Poisson distributed random variable with parameter $\frac{\lambda}{c}$. The same bound applies to the trajectories at a positive height with a positive slope, who are exactly the children of the current resident that have already been born by time $t$. It follows that $K_t/(2N_t)$ tends to 1 in probability, and since $N_t/t$ tends to $\lambda$ almost surely (and thus also in probability) due to the Poisson Law of Large Numbers~\cite[p.~42]{K93}, it follows that $K_t/t$ tends to $2\lambda$ in probability, and hence also almost surely along a subsequence. But we know from Corollary~\ref{cor-speedofkinking} that $K_t/t$ tends to $\E(K_{E_1})/\E(E_1)$ almost surely, and therefore the almost sure subsequential limit $2\lambda$ is actually the almost sure limit of $K_t/t$ as $t\to\infty$.

This simple case already shows that~\eqref{speedofkinking} is in general not true with $E_1$ replaced by $L_1$ everywhere. Indeed, $\E(L_1) = \frac{1}{\lambda} + \frac{1}{c}$, since $L_1$ is almost surely equal to the time $t_1+\frac{1}{c}$ when the first-born trajectory $h_1$ becomes resident. Moreover, $\E(K_{L_1})=\frac{\lambda}{c} + 1$, since before time $L_1$ there is no resident change and thus no slope change that corresponds to the definition of $t \mapsto K_t$, and at time $L_1$ the slope of $h_1$ and all the other Poisson($\lambda/c$) trajectories born between time $t_1$ and $t_1+\frac{1}{c}=L_1$ changes. This implies that $\E(K_{L_1})/\E(L_1) =\lambda $, while we just showed that $\E(K_{E_1})/\E(E_1)=2\lambda$.

As we already anticipated in Remark~\ref{rem-defchanges},~\cite{HGSTW24} studied a variant of the PIT where trajectories have a chance to stay constant equal to zero, which corresponds to rapid extinction of the given mutant subpopulation in the Moran model. There, in the constant-$c$ case, trajectories arrive at rate $\lambda$ but independently with probability $\frac{1}{1+c}$ they stay constant 0 and only with probability $\widetilde\lambda:=\frac{\lambda c}{1+c}$ they follow the dynamics described in Section~\ref{sec-intro}. In that setting, $\lim_{t\to\infty} K_t/t$ would be $2\widetilde\lambda=\frac{2\lambda c}{1+c}$ almost surely, while the speed of adaptation is $\lim_{t\to\infty} F(t)/t = \frac{\lambda c^2}{1+c+\lambda}$ almost surely according to~\cite[Theorem 2.10]{HGSTW24}, a result which in turn originates from~\cite[Section 3.1]{BGPW19}.

\end{remark}

\section{Proofs}\label{sec-proofs}
\subsection{Lower bound construction for the number of kinks: proof of Proposition~\ref{prop-cousinen}}\label{sec-cousinen}
\begin{proof}
A sequence $(\mathbb{H}(\beth^{(n)})_{n\in\mathbb{N}}$ corresponding to the statement of the proposition is given as follows. For simplicity, we use the continued lines representation introduced in Section~\ref{rem-contlines} (which is in bijection with the representation as a system of interacting trajectories; as an illustration we depict the corresponding interacting trajectories for $n=10$ in Figure~\ref{fig-LB}).

\begin{figure}
    \centering
    \includegraphics[width=\linewidth]{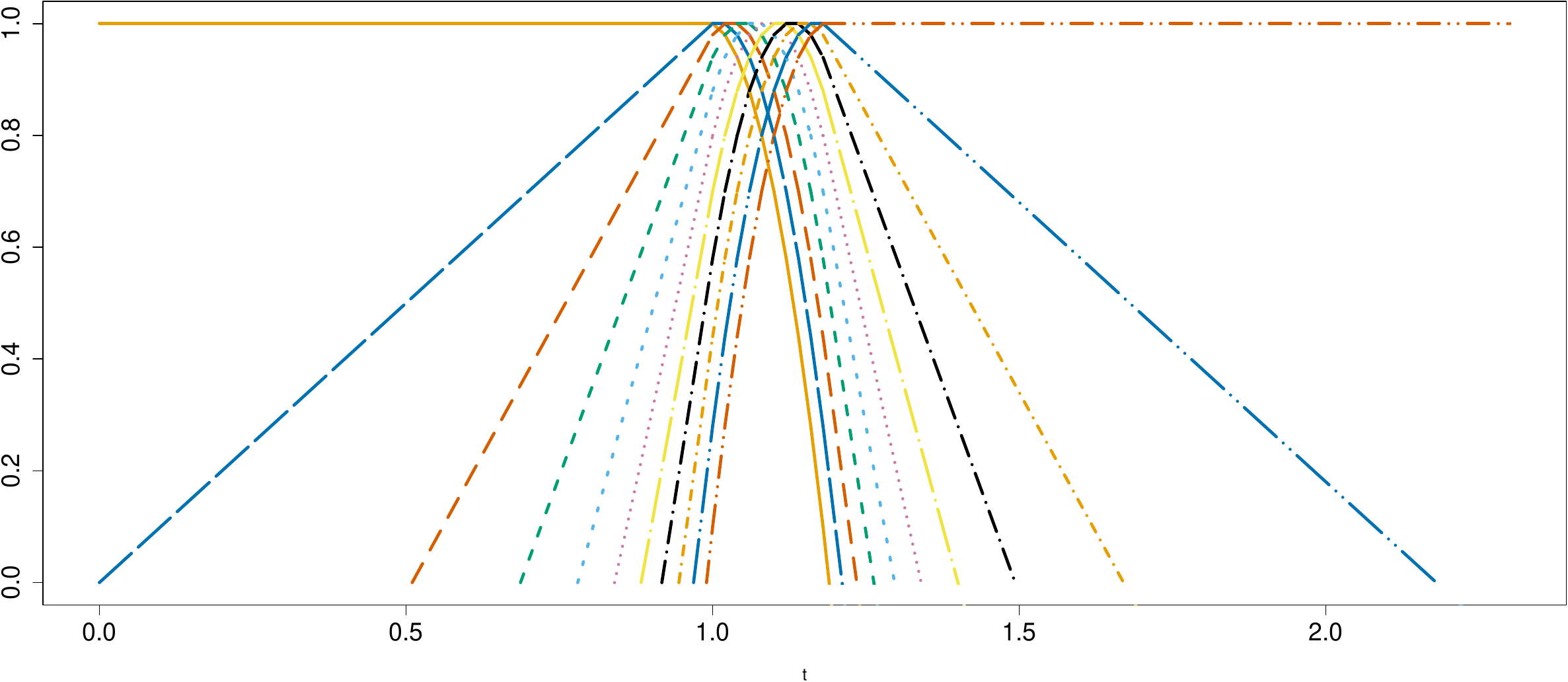}
    \caption{The interacting trajectories corresponding to the {half-lines} $\ell_i$ in the proof of Proposition~\ref{prop-cousinen} for $n=10$. We see that the trajectories $h_1,\ldots,h_{10}$ are all born during the time of residency of the initial resident, in an increasing order w.r.t.\ their initial slopes. Each trajectory becomes resident, and thus the order of time intervals of residencies is the same as the order of their birth times, i.e.\ the reverse of the order of their initial slopes.}\label{fig-LB}
\end{figure}

For $n \in \mathbb{N}$ and $k \in \{0,1,\ldots,n\}$, we define the $k$-th {half-line} $\ell_k$ corresponding to the continued lines representation of $\mathbb{H}(\beth^{(n)})$ via the equations $\ell_0(t)=1$, $t \geq 0$, and for $k \geq 1$,
\[ \ell_k(t) = k \Big( t-(1-\frac{1}{k} + \frac{k-1}{n^2}) \Big) , \qquad t \geq 1-\frac1k+\frac{k-1}{n^2}. \]
First we claim that $\ell_k$ and $\ell_{k+1}$ intersect at time $1+\frac{2k}{n^2}$. Indeed, this is clear for $\ell_0(t)=1$ and $\ell_1(t)=t$. Moreover, for $k \geq 1$, $t$ is the time when $\ell_k$ and $\ell_{k+1}$ intersect if and only if
\[ k \big( t-1+\frac1k - \frac{k-1}{n^2}\big) = (k+1) \big( t-1 +\frac{1}{k+1}-\frac{k}{n^2} \big), \]
which is satisfied if and only if $t = 1+\frac{2k}{n^2}$. Since $\ell_{k+1}$ is steeper than $\ell_k$, this implies that \begin{equation*} \ell_{k+1}(t) > \ell_k(t) \quad \Leftrightarrow \quad t > 1+\frac{2k}{n^2}. \end{equation*}
We now claim that it is even true that
\begin{equation}\label{whoisres} \ell_0 \text{ is resident on } [0,1) \text{ and for $k \geq 1$, } \ell_{k} \text{ is resident on } \big[1+\frac{2(k-1)}{n^2}, 1+\frac{2k}{n^2} \big]. \end{equation}
Indeed, recall that at a given time $t$, the index of the resident is the index of the {half-line} that is at the highest position among all {half-lines} at time $t$ (which implies that the {half-line} has already been started by time $t$, and in the case of a draw, the index with the higher slope is the index of the resident). Now,~\eqref{whoisres} follows inductively from the fact that $0$ is obviously resident at time $0$, type $k$ cannot be resident at time $t<1+\frac{2(k-1)}{n^2}$ because then $\ell_{k-1}(t) > \ell_k(t)$, and it cannot be resident at time $t > 1 + \frac{2k}{n^2}$ either because then $\ell_{k+1}(t)>\ell_k(t)$, but at any time there must be a resident type, so $\ell_k$ must be resident for $\frac{2(k-1)}{n^2} < t < \frac{2k}{n^2}$ and thus by construction also for $t=\frac{2(k-1)}{n^2}$.

Now we claim that the $0$-th trajectory of $\mathbb{H}(\beth^{(n)})$ is still alive at time $1+\frac{2(n-1)}{n^2}$, equivalently, $\ell_0(1+\frac{2(n-1)}{n^2}) > \ell^r(1+\frac{2(n-1)}{n^2})-1 = \ell_{n}(1+\frac{2(n-1)}{n^2}) -1$, where we recall from Section~\ref{rem-contlines} that $\ell^r(t)$ is the height of the current resident {half-line} at time $t$. Indeed, we have that
\[ \ell_0(t) = 1, \qquad \ell_{n}(1+\frac{2(n-1)}{n^2}) = n \big( 1+\frac{2(n-1)}{n^2}-1 +\frac1n - \frac{n-1}{n^2} \big) = n \big( \frac{2n-1}{n^2} \big) = 2-\frac{1}{n}. \]
By construction, for $t > 1+\frac{2(n-1)}{n^2}$, $\ell_0(t)$ is definitely the smallest one among all values $\ell_i(t)$, $i \in \{0,1,\ldots,n\}$. This implies that all other trajectories of $\mathbb{H}(\beth^{(n)})$ are also alive at time $1+\frac{2(n-1)}{n^2}$. Since they are all already born (i.e.\ the corresponding {half-lines} of the continued lines representation are already started) before time $1$, it follows that at all resident change times all trajectories of $\mathbb{H}(\beth^{(n)})$ are alive and thus they change their slopes. This yields $k_n(\beth^{(n)}) = n(n+1) = \Omega(n^2)$.
\end{proof}

\subsection{Proof of Theorem~\ref{theorem-algorithm}}\label{sec-algorithm}

\begin{figure}
\centering
\begin{footnotesize} \includegraphics[scale=0.5]{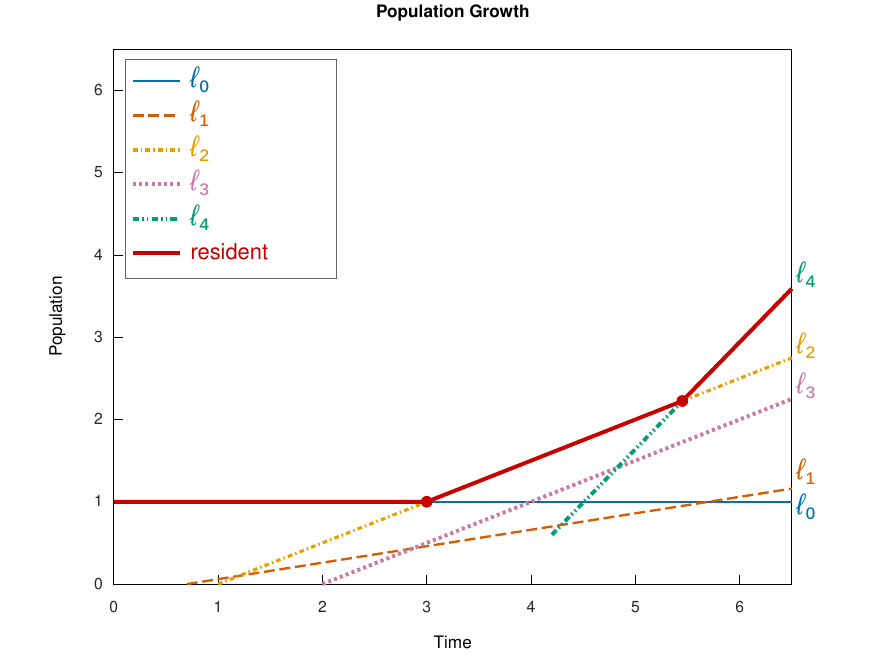}
\includegraphics[scale=0.5]{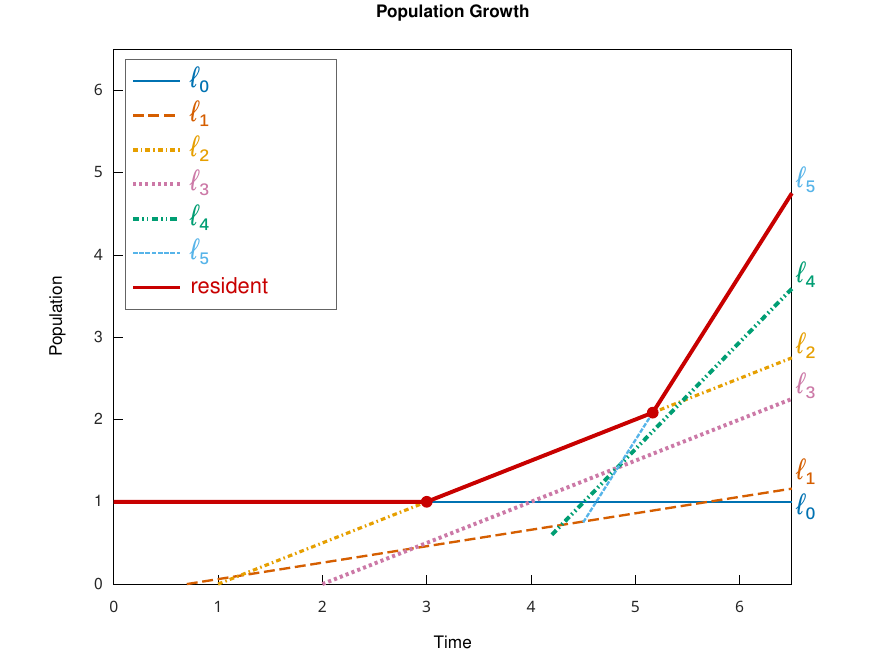} \end{footnotesize}
\caption{\emph{Left}: Continued lines representation of a system of 5 interacting trajectories, with {half-lines} labelled $\ell_0$ through $\ell_4$ in order of birth (see legend). The increasing and convex broken line of the resident {half-line} (i.e.\ always the currently highest {half-line}) is depicted in smooth red (``resident''), with the dots representing the resident changes. $\ell_0$ (the horizontal {half-line} at height 1) is the initial resident trajectory. $\ell_1$, $\ell_2$, $\ell_3$ are born while $\ell_0$ is resident, while the last-born $\ell_4$ is born when $\ell_2$ is resident (middle segment of the resident envelope). Note that each {half-line} starts at the height of the current resident {half-line} minus one. \\
\emph{Right}: Illustration of one of the induction steps of our algorithm in the proof of Theorem~\ref{theorem-algorithm}. The forward linear search identifies the initial height and slope of the last-born $\ell_5$ (resuming from the previous birth time of $\ell_4$). This trajectory is again born at height 1 below the current height of the resident {half-line}. Afterwards, the backward linear search finds the intersection point of $\ell_5$ and the resident {half-line}, and replaces the part of the resident {half-line} after this intersection point by $\ell_5$. In particular, while $\ell_4$ became resident some time after this intersection point in the old system (left image), in the new system it will never become resident since $\ell_5$ intersects it already before it crosses the resident {half-line}.}\label{fig-contlines}
\end{figure}

Recall from Definition~\ref{defdyn} that the continued lines representation consists of {half-lines} with non-decreasing slope, meaning that the segments of half-lines corresponding to the residencies form a convex broken line. The following algorithm is a modified version of the well-known Graham scan (see Algorithm~\ref{alg:residency} below for a pseudocode).

The algorithm processes births one at a time, in increasing order, while maintaining an ordered list of residents. For half-line $i$, the parent is found by linear search from the start of the resident list forwards, and whether it becomes resident is determined by linear search from the end of the resident list backwards.

Recall from Section~\ref{rem-contlines} that in the continued lines representation, the birth of each {half-line} is represented by a pair $(t_i, a_i)$, where $t_i$ is the birth time and $a_i$ is the slope relative to its parent.

A {half-line} is represented by a tuple $(t_i, p_i, b_i, y_i)$, where $t_i$ is the birth time, $p_i$ is the index of the {half-line} of the parent, who is the resident at time $t_i$, $b_i = b_{p_i}+a_i$ is the slope, and $y_i$ is the initial height, which is one less than the height of the parent at time $t_i$. This tuple corresponds to a {half-line} starting at point $(t_i, y_i)$ with slope $b_i$. A residency of a {half-line} is represented by a pair $(s_j, m_j)$, where $s_j$ is the residency start time and $m_j$ is the {half-line} index. A residency starts when its {half-line} intersects the {half-line} of the current resident from below, so the residencies form a convex broken line. The $j$-th residency starts at point $(s_j, y_{m_j}+b_{m_j}(s_j-t_{m_j}))$.

Initially, there is one {half-line} in the system, represented by $(t_0=0, p_0=-1, b_0=0, y_0=1)$, which is the current resident, represented by $(s_0=0, m_0=0)$. When processing the $i$-th birth, the forward linear search finds the parent of {half-line} $i$ in the list of residents, which is the last resident $j$ with $s_j \leq t_i$. The $i$-th {half-line} is represented by $(t_i, p_i=m_j, b_i=(b_{p_i}+a_i), y_i=(y_{p_i}+b_{p_i}(t_i-t_{p_i})-1))$. Since the input is ordered by $t_i$, the search for the parent can start from the parent of the previous {half-line}.

If {half-line} $i$ is below the convex broken line of residents, it does not become a resident. Otherwise, there is an intersection point, which is found by backward linear search. Starting from the last resident in the list, it discards all residency segments from the previous solution below the new {half-line} until the intersection point $(x,y)$ is found and in that case the new residency represented by $(x, i)$ is added to the end.

Since both linear searches pass each {half-line} at most once, the algorithm has runtime $O(n)$. Note here that as a byproduct, this algorithm also computed the half-lines represented by the $(t_i, p_i, b_i, y_i)$ tuples in the continued lines representation of the system.

The lower bound construction appearing in the proof of Proposition~\ref{prop-cousinen} (see Section~\ref{sec-cousinen}) shows that for all $n\in\N$ there exists $\beth^{(n)}$ such that $\mathbb{H}(\beth^{(n)})$ exhibits $n$ resident changes. Therefore, the size of the output of any algorithm computing $t \mapsto \varrho^{(n)}(t)$ is $\Omega(n)$ in the worst case, and hence the runtime of any such algorithm is also $\Omega(n)$.

We conclude that the runtime complexity is indeed $\Theta(n)$, which completes the proof of Theorem~\ref{theorem-algorithm}.

\algrenewcommand\algorithmicrequire{\textbf{Input:}}
\algrenewcommand\algorithmicensure{\textbf{Output:}}
\begin{algorithm}[H]
\caption{Computation of the residency sequence $t \mapsto \varrho^{(\beth)}(t)$}\label{alg:residency}
\begin{algorithmic}[1]
\Require Nondegenerate, ordered $\beth = ((t_i, a_i))_{i=1}^n$
\Ensure Array ${\mathcal{R}} = ((s_j, m_j))_j$ representing the residencies;
\Statex \hspace{2.45em} Array ${\mathcal{HL}} = ((t_i, p_i, b_i, y_i))_{i=0}^n$ representing the {half-lines}
\State ${\mathcal{HL}}[0] \gets (t_0=0,\; p_0=-1,\; b_0{=}0,\; y_0{=}1)$
\State ${\mathcal{R}}[0] \gets (s_0{=}0,\; m_0{=}0)$ \Comment{half-line $0$ is the initial resident starting at $t=0$}
\State $j \gets 0$ \Comment{forward index: last residency with start time $\leq t_i$}
\State $k \gets 0$ \Comment{end index: last valid entry in ${\mathcal{R}}$}
\For{$i \gets 1$ \textbf{to} $n$}
  \State $(t_i, a_i) \gets \beth[i]$
  \While{$j < k$ \textbf{and} $s_{j+1} \leq t_i$} \Comment{Forward search: last residency before $t_i$}
    \State $j \gets j + 1$
  \EndWhile
  \State $p_i \gets m_j$ \Comment{half-line $p_i$ is the resident at time $t_i$, i.e.\ the parent of half-line $i$}
  \State $b_i \gets b_{p_i} + a_i$ \Comment{absolute slope of half-line $i$}
  \State $y_i \gets y_{p_i} + b_{p_i}(t_i - t_{p_i}) - 1$ \Comment{initial height: one below parent at $t_i$}
  \State ${\mathcal{HL}}[i] \gets (t_i, p_i, b_i, y_i)$
  \While{$k \geq 0$ \textbf{and} $b_i > b_{m_k}$} \Comment{Backward search: insert half-line $i$ into ${\mathcal{R}}$}
    \State $x \gets \bigl((y_{m_k} - b_{m_k} t_{m_k}) - (y_i - b_i t_i)\bigr)\,/\,(b_i - b_{m_k})$ \Comment{intersection time}
    \If{$x < s_k$}
      \State $k \gets k - 1$ \Comment{half-line $m_k$ is dominated by half-line $i$ and does not become a resident}
      \State \textbf{continue}
    \EndIf
    \State $k \gets k + 1$
    \State ${\mathcal{R}}[k] \gets (x,\; i)$ \Comment{half-line $i$ becomes resident at time $x$}
    \State \textbf{break}
  \EndWhile
\EndFor
\State \Return ${\mathcal{R}}$, ${\mathcal{HL}}$
\end{algorithmic}
\end{algorithm}

\subsection{Proof of Corollary~\ref{cor-k}}\label{sec-k}
The key step of the proof of Corollary~\ref{cor-k} is the following lemma.
\begin{lem}\label{lem-k}
Let $k>0$. For any nondegenerate, ordered (and not necessarily finite) $\beth$ of size $n$ such that each initial slope $a_i$ lies in $ (0,k]$, for all $t \geq 0$ and $i \geq 0$ we have $v_i(t) \in (-2k, k]$.
\end{lem}
\begin{proof}
Throughout the proof, we will use the following slightly generalized version of~\cite[Lemma 2.5]{HGSTW24} (whose proof is analogous to the one of the original lemma). For all $i \neq j$, recalling that
\[ e_i = \inf \{ t > t_i \colon h_i(t)=0 \}, \]
i.e.\ $e_i$ is the first time when the trajectory $h_i$ hits $0$ again after its birth time $t_i$, we have that
\[ t\mapsto v_i(t) - v_j(t) \text{ is constant on } [\max \{t_i, t_j \}, \min \{e_i, e_j \}). \numberthis\label{lemma4.5} \]
Now assume for a contradiction that for some $t \geq 0$ and $i \geq 0$ we have $v_i(t) \leq -2k$. Put
\[ \tau = \inf \{ s \geq 0 \colon v_i(s) \leq -2k \}. \]
It is clear that our assumption is equivalent to the assumption that $\tau<\infty$.
Recall that $h_i$ is continuous and piecewise linear and hence its right derivative $v_i$ is right-continuous. Thus, if $\tau<\infty$, then $v_i(\tau) \leq -2k$ holds, and thanks to the non-degeneracy assumption, there is a unique $j \in \mathbb{N}_0$ such that the $j$-th trajectory becomes resident at time $\tau$ (see Figure~\ref{fig-contradiction} for an illustration of the proof from this point on).

Using that $v_j(s) \leq k$ for all $s$ and that the $j$-th trajectory becomes resident at time $\tau$, it follows that
\[ \tau \geq t_j + \frac1k. \numberthis\label{taularge} \]
 Moreover, by construction, under the assumption $\tau<\infty$, $v_i(t_j)-v_i(\tau)$ equals the increment of the resident fitness between times $t_j$ and $\tau$ (see~\eqref{lemma4.5} or also~\cite[Lemma 2.5]{HGSTW24}), which is precisely equal to $a_j \leq k$. It follows that $v_i(t_j) \leq -k$. Thus, if $\tau<\infty$, then $v_i(s)\leq -k$ for all $s \in [t_j,\tau)$. Therefore,
 using~\eqref{taularge} we have
\[ h_i(\tau)=h_i(t_j) + \int_{t_j}^{\tau} v_i(s) \mathrm{d} s \leq 1 - k(\tau-t_j) \leq 1-k \cdot \frac1k = 0. \]
This is a contradiction because if $\tau<\infty$, then $h_i(\tau)$ must be positive so that the slope $v_i$ can become $-2k$ or less at time $\tau$. Hence, it follows that $\tau=\infty$.
This finishes the proof of the lemma.

\end{proof}

Now we can carry out the proof of Corollary~\ref{cor-k}.

\begin{proof}[Proof of Corollary~\ref{cor-k}]
Note that if each $a_i$ is integer-valued, then so is $v_i(t)$ for all $0 \leq i < n+1$ and $t \geq 0$.
Since at any resident change time, the slope of any trajectory that is currently at a positive height decreases by at least one and, thanks to Lemma~\ref{lem-k}, the slope of any trajectory at any time ranges between $-(2k-1)$ and $k$, it follows that the slope of any trajectory can change at most $3k-1$ times while the trajectory is at a strictly positive height. In Section~\ref{rem-givent} we describe how to compute the height functions $t \mapsto h_i(t)$ for all trajectories in $O(kn)$ time, assuming each trajectory undergoes at most $O(k)$ kinks, which completes the proof of this corollary.
\end{proof}
\begin{remark}\label{rem-k} \phantom{Itt egy új sort csinálunk. \\}
\begin{enumerate}
\item If the initial slopes $a_i$ can take all values in $[k]$, then there exists a nondegenerate $\beth$ such that the corresponding $(h_i)_i$ exhibits some $i \geq 0$ and $t \geq 0$ such that $v_i(t) = -(2k-1)$, which is the smallest possible slope value according to Lemma~\ref{lem-k}. We claim that this holds for $i=0$ if $a_1 = k-1$, $a_2=k$, $t_1 + \frac1{k-1} < t_2 < t_1 +\frac{2}{k-1}-\frac{1}{k}$, and $t_3 > t_2+\frac{1}{k}$. Indeed, in this case, at time $t_1+\frac1{k-1}$, the first trajectory becomes resident and the slope of the $0$-th trajectory drops by $v_{1}(t_1+\frac1{k-1}-)=\frac{1}{k-1}$ from $0$, i.e.\ it becomes $-(k-1)$, and this slope stays $-(k-1)$ until either the second trajectory becomes resident, which happens at time $t_2+\frac1{k}$ (since the third trajectory is only born after this time), or until the $0$-th trajectory hits zero again, which happens at time $t_1+\frac2{k-1} = \big( t_1 + \frac{1}{k-1}\big) + \frac1{k-1}$. Now, since
\[ t_2 + \frac{1}{k} < t_1 + \frac2{k-1} - \frac{1}{k} + \frac{1}{k} = t_1+\frac2{k-1}, \]
until time $t_2 + \frac{1}{k}$ the $0$-th trajectory does not yet reach height zero again. Consequently, at this time, its slope drops by $v_{2}(t_2+\frac1k - ) = k$, i.e.\ it becomes $-(2k-1)$.
\item In the case of the PIT, if $\lambda>0$ and the support of the measure $\gamma$
contains both $k$ and $k-1$ (i.e. $\mathbb{P}(A_i=k)$ and $\mathbb{P}(A_i=k-1)$ are both positive),
the event that the scenario described in the previous paragraph occurs with $(t_i,a_i)=(T_i,A_i)$ and $h_i=H_i$, $i=1,2,3$, has positive probability. Using a restart argument based on e.g.\ the renewal extinction times defined below in Section~\ref{sec-average}, it follows that for such $\gamma$, the $\mathrm{PIT}(\lambda,\gamma)$ exhibits infinitely many trajectories whose slope reaches $-(2k-1)$, almost surely. The same is clearly false e.g.\ if $k$ is even and the support of $\gamma$ is contained in $\{ 2,4,\ldots,k\}$.
\item While Lemma~\ref{lem-k} guarantees that trajectories can never have slopes less than or equal to $-2k$ when all initial slopes $a_i$ are in $(0,k]$, slopes in $(-(2k-1),-2k)$ can be obtained for certain choices of the possible values of $a_i$ (using an analogue of the construction in the first paragraph of the present remark). Moreover, since the slope decrement due to one kink can be arbitrarily small, having a finite lower bound on the minimal slope does not imply in general that the number of kinks of a given trajectory is bounded by a constant times the absolute value of the minimum possible slope. This is an obstacle to generalizing Corollary~\ref{cor-k} to this case.
\item
An analogue of Corollary~\ref{cor-k} holds whenever the possible initial slopes $a_i$ are in $\{ \alpha, 2 \alpha, \ldots, k\alpha \}$ for some $\alpha>0$.
\end{enumerate}
\end{remark}

\begin{figure}
\centering
\scalebox{0.85}{
\begin{tikzpicture}
  \draw[scale=3.7, black, ->] (0.9, 0) -- (3.95, 0) node[right] {time};
  \draw[scale=3.7, black, ->] (0.9, -0.0125) -- (0.9, 1.075) node[above] {height};
  \draw[scale=3.7,black,thick] (0.87,1) node[left] {$1$};
    \draw[scale=3.7,black,thick] (0.87,0) node[left] {$0$};
  \draw[scale=3.7, blue,thick] (1.65,0.17) node[above] {$v_j(t_j)=a_j \leq k$};
  \draw[scale=3.7, blue,thick] (1.65,0.02) node[above] {$h_j(t_j) =0$};
  \draw[scale=3.7, blue,thick] (1.75,0.4) node[above] {$v_j(t) \leq k$ on $(t_j,\tau)$};
  \draw[scale=3.7, blue,thick] (2.05,-0.00725) node[below] {$ t_j$};
  \draw[scale=3.7, red,thick] (2.58,-0.00725) node[below] {$ \tau$};
  \draw[scale=3.7, blue,thick] (2.7,0.95) node[below] {$h_j(\tau)=1$};
  \draw[scale=3.7, blue,thick] (2.7,0.8) node[below] {$v_j(\tau)=0$};
  \draw[scale=3.7, red,thick] (1.75,0.95) node[below] {$v_i(t_j) \leq -k$};
  \draw[scale=3.7, red,thick] (1.75,0.8) node[below] {$h_i(t_j) \leq 1$};
  \draw[scale=3.7, red,thick] (3.09,0.55) node[below] {$v_i \leq -k$ on $(t_j,\tau)$ $\Rightarrow$ $h_i(\tau)\leq 0$ {\normalcolor \Huge \Lightning}};
  \draw[scale=3.7, red,thick] (2.85,-0.05) node[below] {$v_i(\tau)\leq -2k$};
    \draw[scale=3.7, domain=2:2.5, smooth, variable=\x, traj1] plot ({\x}, {2*(\x-2)});
    \draw[scale=3.7, domain=2.5:3, smooth, variable=\x, traj1] plot ({\x}, {1});
  \draw[scale=3.7, domain=2:2.5, smooth, variable=\x, traj2] plot ({\x}, {1-2*(\x-2)});
   \draw[scale=3.7, domain=2.5:2.57, variable=\x, hyp2] plot ({\x}, {-4*(\x-2.5)});
\end{tikzpicture}}
\caption{Illustration of the proof of Lemma~\ref{lem-k}. Since $v_i(\tau) \leq -2 k$ and $v_j$ decreases from $a_j \leq k$ to $0$ between time $t_j$ and $\tau$, $v_i(t_j)$ cannot be larger than $-k$ (this holds thanks to~\eqref{lemma4.5}). But this implies that it takes $h_i$ at most as much time to get from height at most $1$ to height $0$ as it takes $h_j$ to get from height $0$ to $1$, which implies $h_i(\tau) \leq 0$, a contradiction from which it follows that $\tau$ must be infinite.}\label{fig-contradiction}
\end{figure}

\subsection{Proof of Proposition~\ref{lemma-average} and Corollary~\ref{cor-speedofkinking}}\label{sec-average}

\begin{proof}[Proof of Proposition~\ref{lemma-average}]
Recall the renewal extinction times $E_i$ from~\eqref{Ei}. Note that for $i \geq 1$, on the time interval
$(E_{i-1},E_i]$ only trajectories born in $(E_{i-1},E_i]$ can become resident,
and only the height functions of the same trajectories plus the one of the trajectory that was resident at time $E_{i-1}$ can suffer any kinks on this time interval due to resident changes.
Indeed, any other trajectory born on $[0,E_{i-1}]$ goes extinct by the time $E_{i-1}$. Hence, we have
\[ K_{E_i}-K_{E_{i-1}} \leq (N_{E_i}-N_{E_{i-1}})(N_{E_i}-N_{E_{i-1}}+1), \numberthis\label{KN2} \] and for $t \geq 0$, putting $\eta_t = |\{ i \geq 1 \colon E_i \leq t \}|=\sup \{ i \geq 0 \colon E_i \leq t \}$, we have $E_{\eta_t} \leq t < E_{\eta_t+1}$. Note also that $\eta_t \leq N_t$ for all $t$ because at each $E_i$ a trajectory born before $E_i$ becomes resident (and one trajectory can only become resident at most once).

It can be proven analogously to the proof of~\cite[Lemma 5.2]{HGSTW24} that $E_{i}-E_{i-1}$ is stochastically dominated by constant times a geometric random variable (see Appendix~\ref{sec-Nmoment} for more details). Moreover, the random variables $(E_i-E_{i-1})_{i\geq 1}$ are i.i.d.\ (in other words, $(E_n)_{n\in\mathbb{N}}$ is a renewal process). They have a finite first moment $m_1:=\E(E_1)$.
This implies that
\[ \begin{aligned} \E(K_{E_i}-K_{E_{i-1}}) & \leq \E((N_{E_i}-N_{E_{i-1}})(N_{E_i}-N_{E_{i-1}}+1)) = :\mathfrak m.
\end{aligned} \numberthis\label{KE}
\]
It is clear that $\mathfrak m>0$ because $N_{E_i}-N_{E_{i-1}} \geq 1$ for any $i$. The fact that $\mathfrak m<\infty$ is an immediate consequence of the following assertion.
\begin{lem}\label{lemma-NE1}
For any $\gamma$, it holds that $\E(N_{E_1}^2)<\infty$.
\end{lem}
The proof of this lemma follows using an argument analogous to~\cite[Section 5]{HGSTW24}. In Appendix~\ref{sec-Nmoment}, we will sketch this argument.
Given Lemma~\ref{lemma-NE1}, for $n\geq 1$ we have
\[ \E(K_{E_n}) = \sum_{i=1}^n \E(K_{E_i}-K_{E_{i-1}}) \leq n \mathfrak m. \numberthis\label{nrenewals} \]
The same arguments that lead to~\eqref{KN2} also yield the (much less tight) bound
\[K_t \leq N_t(N_t +1) \numberthis\label{maxnumberofkinks} \]
almost surely for all $t \geq 0$.
Since $(E_n)_{n\in\N}$ is a renewal process with increments having finite mean $m_1$, we have
\[ \lim_{n\to\infty} \frac{E_n}{n} = m_1 \]
almost surely thanks to the strong law of large numbers, and
\[ \lim_{t\to\infty} \frac{\eta_t}{t} = \frac{1}{m_1} \]
almost surely (cf.~\cite[Theorem 2.5.10]{EKM97} for the latter assertion). Using Hölder's inequality,~\eqref{nrenewals} and and~\eqref{maxnumberofkinks}, we conclude that for $\delta>0$, for all $t$ sufficiently large
\[
\begin{aligned}
    \E(K_{t}) & \leq  \E(K_{E_{\eta_t + 1}} \mathds 1_{\{\eta_t+1 \leq (\frac{1}{m_1} + \delta)t \}}) + \E(K_{t} \mathds 1_{\{\eta_t +1 > (\frac{1}{m_1} + \delta)t \}})
    \\ & \leq  \E(K_{E_{\lceil (\frac{1}{m_1}+\delta )t \rceil}}) + \sqrt{\E(K_t^2)}\sqrt{\E\Big( (\mathds 1_{\{\eta_t +1 > (\frac{1}{m_1} + \delta)t \}})^2\Big)}
    \\ & \leq (\frac{1}{m_1}+\delta + 1) t \mathfrak m +  \sqrt{\E(K_t^2)}\sqrt{\P\Big( \eta_t +1 > (\frac{1}{m_1} + \delta)t \Big)}
    \\ & \leq (\frac{1}{m_1}+\delta + 1) t \mathfrak m +  \sqrt{\underbrace{\E(N_t^4+2N_t^3+N_t^2)}_{\in O(t^4)}}\sqrt{\P\Big( \eta_t > (\frac{1}{m_1} + \delta)t - 1 \Big)}
    \\ & \leq (\frac{1}{m_1}+\delta + 1) t \mathfrak m + O(t^2) \sqrt{\P\Big( \eta_t > (\frac{1}{m_1} + \delta)t - 1 \Big)}.
\end{aligned}
\numberthis\label{deltasplit}
\]
Here, the assertion that $\E(N_t^4+2N_t^3+N_t^2) \in O(t^4)$ follows from the well-known fact that a Poisson distributed random variable $X$ with parameter $\lambda$ satisfies $\E(X(X-1)\ldots(X-k+1))=\lambda^k$ for any $k \geq 1$, from which one can easily derive that $\E(X^4) = (1+o(1)) \lambda^4$ as $\lambda \to \infty$.

To treat the last term on the right-hand side of~\eqref{deltasplit}, we use standard large-deviation estimates. First
note that there exists $\delta'>0$ not depending on $t$ such that
\[ \Big\{ \eta_t > \big( \frac{1}{m_1} + \delta \big) t -1 \Big\} = \Big\{ E_{\lceil (\frac{1}{m_1} + \delta)t \rceil - 1 } \leq t \Big\} \subseteq \Big\{ E_{\lceil (\frac{1}{m_1} + \delta)t \rceil - 1} \leq (1-\delta') m_1 \big( \lceil (\frac{1}{m_1} + \delta)t \rceil -1 \big) \Big\}. \]
Since $\Lambda(\alpha):=\log \E(\e^{\alpha E_1})<\infty$ for $\alpha>0$ small (which also follows from the stochastic domination by constant times a geometric random variable), Cramér's theorem~\cite[Section 2.2]{DZ98} implies that if $\delta'$ is sufficiently small, then
\[ \limsup_{n\to\infty} \frac{1}{n} \log \P \Big( E_{n} \leq (1-\delta') m_1 n \Big) \leq -\inf_{ t \leq (1-\delta')m_1} \Lambda^*(t) =: -I_0 <0, \numberthis\label{ratefunctionnegative} \]
where the rate function $\Lambda^* \colon \R \to [0,\infty]$ is given as
\[ \Lambda^*(a) = \sup_{\alpha \in \R} (\alpha a - \Lambda(\alpha)), \]
and together with the fact that $\Lambda'(0)=m_1$ and $\Lambda^*(m_1)=0$, the strict negativity in~\eqref{ratefunctionnegative} follows from the assertion of~\cite[Exercise 2.2.24]{DZ98} for $\delta'>0$ small enough.
We conclude that for all sufficiently large $t$,
\[ \sqrt{\P\Big( \eta_t > (\frac{1}{m_1} + \delta)t \Big)} \leq \e^{- \frac{1}{2} I_0 ( \lceil (\frac{1}{m_1} + \delta)t \rceil - 1)} \leq \e^{- \frac{1}{2} I_0 (\frac{1}{m_1} + \delta/2)t }. \]
Thus, the right-hand side of~\eqref{deltasplit} is bounded from above by
\[ (\frac{1}{m_1}+\delta + 1) t \mathfrak m + O(t^2) \e^{- \frac{1}{2} I_0 (\frac{1}{m_1} + \delta/2)t} \in O(t), \]
as wanted.
\end{proof}

\begin{proof}[Proof of Corollary~\ref{cor-speedofkinking}]

This proof is analogous to the one of~\cite[Theorem 2.8]{HGSTW24}.
We note that
\[ \widehat K_t = \sum_{i \geq 0 \colon E_i \leq t } K_{E_i}-K_{E_{i-1}} = \sum_{i \geq 0} K_{E_i} \mathds 1_{[E_{i},E_{i+1})}(t) \]
is a renewal reward process with renewal times $E_i$, i.e.\ it has i.i.d.\ increments at the renewal times $E_i$ that are independent of the $E_i$'s, and thus the corollary is a quick consequence of
the law of large numbers. For the reader's convenience we recall the argument.
For $t \geq 0$ let $m(t)$ be such that $E_{m(t)} \leq t < E_{m(t)+1}$. Then we have
\[ \frac{K_{E_{m(t)}}/m(t)}{E_{m(t)+1}/m(t)} \leq \frac{K_t}{t} \leq \frac{K_{E_{m(t)+1}}/m(t)}{E_{m(t)}/m(t)}, \numberthis\label{Kbounds} \]
analogously to~\cite[Equation (5.9)]{HGSTW24}.
Since
\begin{itemize}
    \item $m(t) \to \infty$ as $t \to \infty$,
    \item $E_m$ is a sum of $m$ i.i.d.\ copies of $E_1$ which have finite (and clearly strictly positive) expectation (cf.\ the proof of Proposition~\ref{lemma-average}),
    \item and $K_{E_m}=\widehat K_{E_m}$ is a sum of $m$ i.i.d.\ copies of $K_{E_1}$,
\end{itemize}
both the left- and the right-hand side of~\eqref{Kbounds} converge almost surely to $\E(K_{E_1})/\E(E_1)$. Hence, in order to finish the proof of the corollary, it suffices to show that $\E(K_{E_1}) \in (0,\infty)$. The fact that $\E(K_{E_1})<\infty$ follows from~\eqref{KE} for $i=1$. Finally, $\E(K_{E_1})$ is positive because $K_{E_1}$ is almost surely at least 2. Indeed, if the input of the PIT is nondegenerate, then between times $0$ and $E_1$, the slope of the initial resident trajectory must become negative, while the slope of the resident at time $E_1$ must become zero.
\end{proof}

\section{Discussion}\label{sec-discussion}

\subsection{Ancestral relations and fixations}\label{rem-fix}
According to~\cite[Section 6.1]{HGSTW24}, for $i \geq 0$ we say that $h_i$ \emph{fixes} (or \emph{reaches fixation}) if there exists $t_0>0$ such that for all $t > t_0$, the trajectory $h_i$ is an ancestor (i.e., parent, or parent of parent etc.) of all trajectories $h_j$ such that $h_j(t)>0$.
If the $i$-th trajectory is alone in the system with its parent, $h_i$ reaches height 1 at time $t_i + \frac1{a_i}$ and immediately fixes. But when other trajectories are present, due to clonal interference with them, this is not always the case: The slope of the $i$-th trajectory may change due to resident changes after time $t_i$.

Given our algorithm in the proof of Theorem~\ref{theorem-algorithm} for the residents, for finite $n$ we can now determine in linear time which trajectories fix. Indeed, the trajectories that fix are precisely the ultimate resident and its ancestors, and this ancestral line can be found in $O(n)$ time
using linear search. We can search for the parent of the ultimate resident starting from the end of the list in the output of the algorithm, and then we can search for the parent of this parent continuing backwards etc.

\subsection{Computing the height function}\label{rem-givent}

Recall that $t \mapsto \varrho^{(\beth)}(t)$ specifies the resident type at time $t \geq 0$, and its array representation can be computed in a runtime of $O(n)$ as in Section~\ref{sec-algorithm}. This algorithm also produces all trajectories in the continued lines representation, each given by a tuple $(t_i, p_i, b_i, y_i)$ with birth time $t_i$, parent index $p_i$, slope $b_i$, and initial height $y_i$. Using the output of this algorithm, the following computations can be made.

The $i$-th {half-line} in the continued lines representation is $\ell_i(t) = y_i + b_i(t - t_i)$ for $t \geq t_i$. From this, in the original system of interacting trajectories, we have $h_i(t) = \ell_i(t)- (\ell_{p_i}(t)-1)$ if this quantity is positive and $h_i(t)=0$ otherwise. This can be evaluated at any given $t \geq 0$ in $O(1)$.

Moreover, the entire piecewise linear trajectory $t \mapsto h_i(t)$ can be determined in $O(n)$ time as follows: at time $t_i$ it starts with slope $a_i$, and afterwards at each resident change time $s_{j}$, its right slope $v_i(\cdot )$ is reduced by the slope increment of the resident at this time. This way, $h_i(s_j)$ can be computed iteratively from the residency array for all $s_j > t_i$; when we first obtain a negative number, we check when precisely between $s_{j-1}$ and $s_j$ the trajectory $h_i$ hits zero, and from that time on we let $h_i$ be identically zero.

Finally, to compute $h_i(t)$ for all $i$, one may apply the previous procedure independently for each $i$, resulting in an $O(n^2)$ algorithm. If it is known in advance that each trajectory undergoes at most $O(k)$ kinks, a more efficient approach is to merge the residency array (keyed by residency start time) with the trajectory array (keyed by birth time) and, at each residency change, keep track of and update only the trajectories currently alive. As each trajectory experiences at most $O(k)$ kinks, this yields an algorithm of runtime $O(kn)$.

\subsection{Notes on the case of an unordered \texorpdfstring{$\beth$}{beth}}\label{sec-unordered}
In the case of an unordered $\beth$, the algorithm from Section~\ref{sec-algorithm} can be preceded by a sorting step, yielding an $O(n \log n)$ construction.

For a lower bound, consider an unordered list of pairwise distinct integers $\{x_1, \dots, x_n\}$. For $\beth = \{(x_1, 2), \dots, (x_n, 2)\}$, all trajectories become resident in increasing order of $x_i$. Thus, any algorithm that computes the array representation of the piecewise constant functions $t \mapsto \varrho^{(\beth)}(t)$ or $t \mapsto f(t)$ implicitly sorts the $x_i$, which in the comparison model requires $\Omega(n \log n)$ time.

Furthermore, the parent of each trajectory $(x_i, 2)$ is the one immediately preceding it in the sorted order of the $x_i$. Selecting the maximum $x_i$ and repeatedly taking the parent recovers the entire parent–child chain, yielding the sorted sequence in strictly decreasing order of the $x_i$. This means that computing the parent-child relationships between the trajectories is at least as hard as sorting. In the comparison model this establishes a lower bound of $\Omega(n \log n)$.

These lower bounds match the upper bound obtained by first sorting $\beth$ and then applying the algorithm from Section~\ref{sec-algorithm}.

\subsection{Extension of the algorithm
to interacting trajectories with jumps}\label{sec-b} Consider a variant of the system of interacting trajectories where trajectories start from a fixed height $b \in (0,1)$ instead of 0 upon birth and jump to zero after reaching height $b$ from above again (and otherwise they follow the dynamics described in Section~\ref{sec-model}). At the end of Section~\ref{rem-modsel}, we will put this variant into population-genetic context.

In the continued lines representation of this variant of the system of interacting trajectories, {half-lines} are born $1-b$ (instead of $1$) below the current height of the resident {half-line}, and reaching the current resident height minus ($1-b$) again corresponds to becoming 0 again.
It is clear that our algorithm for determining the resident fitness remains valid for this variant of the continued lines representation, still with linear runtime and with the height of newborn trajectories being computed accordingly.

\subsection{Background from population genetics}\label{rem-modsel}

It was shown in~\cite{HGSTW24} that the system of Poissonian interacting trajectories arises as a scaling limit of a Moran model with mutation and selection. This model is a continuous-time Markov chain with fixed total population size $N$. Assume that at time 0, all $N$ individuals belong to the initial resident population and have the same fitness value 0. Beneficial mutations occur in a randomly chosen individual at time $T_i \log N$ (with $T_i$ corresponding to the PIT), with a positive fitness advantage equal to $A_i$, for all $1 \leq i < n+1$. Since $A_i$ does not depend on $N$, one speaks about \emph{strong selection}.
Mutant subpopulations tend to either go extinct rapidly or they grow approximately exponentially. In the latter case, the base-$N$ logarithm of their population size grows approximately linearly and takes $\Theta(\log N)$ time to reach the vicinity of 1.

The per generation mutation rate is $\Theta(1/\log N)$, which corresponds precisely to the so-called Gerrish--Lenski mutation regime.
Speeding up time by a factor of $\log N$, we obtain approximately the system $\mathbb{H}(\beth)$; the corresponding convergence result is~\cite[Theorem 2.7]{HGSTW24}. Here, in the limit, the height $h_i(t)$ of a given trajectory at a given time $t$ represents the $N$-base logarithm of the size of a subpopulation of genetically identical individuals in the Moran model.

Consider now a variant of the Moran model where for some $b \in (0,1)$ fixed, the random fitness advantages $A_i$ are multiplied by the vanishing factor $N^{-b}$ depending on the total population size $N$. This is called the case of \emph{moderate selection}, where mutants typically survive initial fluctuations with probability of order $N^{-b}$, cf.\ \cite{BGPW21i,BGPW21ii}. It was conjectured in~\cite[Section 7.2]{HGSTW24} that if one still chooses the per generation mutation rate as $\Theta(1/\log N)$ and speeds up time by a factor of $N^b \log N$, the system of rescaled logarithmic frequencies of mutants surviving the initial fluctuations converges to the variant of the system of (Poissonian) interacting trajectories with $b$-jumps that was sketched in Section~\ref{sec-b}.

\appendix

\section{Appendix: sketch of proof of Lemma~\ref{lemma-NE1}}\label{sec-Nmoment}
First, we already anticipated in Section~\ref{sec-k} that $E_1$ is stochastically dominated by constant times a geometrically distributed random variable. Let us now provide some details of the proof of this assertion, which are necessary in order to explain why $N_{E_1}$ has a finite second moment. We claim that if $i \in \mathbb{N}$ is such that for all $i' \neq i$ we have $T_{i'} \notin [T_i - \frac{2}{A_i}, T_i + \frac{4}{A_i}]$, then there is a renewal extinction time in $[T_i - \frac{2}{A_i}, T_i + \frac{4}{A_i}]$. If $i$ has this property, then we will call $i$ \emph{favourable}.

Indeed, it was shown in~\cite[Section 5.2]{HGSTW24} that if $i\in \mathbb{N}$ is such that for all $\mathbb{N}_0 \ni i' \neq i$ we have $T_{i'} \notin [T_i - \frac{2}{A_i}, T_i + \frac{2}{A_i}]$, then the $i$-th trajectory $H_i$ becomes resident (i.e.\ it reaches height $1$) at some time $R$ that is not larger than $T_i + \frac{2}{A_i}$, and this resident change is solitary (cf.\ Remark~\ref{rem-analogy}), i.e.\ all trajectories but the $i$-th one have a nonpositive slope right after this resident change. Moreover, it was also shown there that $V_i(R-)=H_i'(R-) \geq A_i/2$ and that $V_{i'}(T_i) \leq A_i/2$ for any $i'<i$. (In words, the ultimate slope of the $i$-th trajectory before this resident change is at least $A_i/2$, and any trajectory $H_{i'}$ born before time $T_i$ has slope at most $A_i/2$ at time $T_i$.) Hence, it follows from~\eqref{lemma4.5} that
\[ V_i(R-) - V_{i'}(R-) = V_i(T_i)-V_{i'}(T_i) \geq A_i - \frac{A_i}{2} = \frac{A_i}{2}. \]
Now, let us note that according to Definition~\ref{defdyn} this implies that for all such $i'$ we have
\[ V_{i'}(R) \leq V_i(R)-\frac{A_i}{2} = 0 -\frac{A_i}{2} = -\frac{A_i}{2}. \]
There exists at least one such index $i'$, namely the previous resident before the $i$-th trajectory. Moreover, since the slope of any trajectory of the PIT is nonincreasing between the birth time and the extinction time of the trajectory, it follows that these trajectories hit zero at time not later than $T_i + \frac{4}{A_i}$. We conclude that if there is no trajectory born in $[T_i + \frac{2}{A_i}, T_i + \frac{4}{A_i}]$ either, then all of these trajectories hit zero by time $T_i + \frac{4}{A_i}$ and the time when the last one hits zero is a renewal extinction time. This implies the claim.

Next, choose $a_0>0$ such that $q:=\gamma([a_0,\infty))=\mathbb{P}(A_1 \geq a_0)>0$. For $\ell \in \mathbb{N}$, let us call $\ell$ \emph{good} if there exists precisely one $i \in \mathbb{N}$ such that $T_i \in [\frac{7(\ell-1)+2}{a_0},\frac{7(\ell-1)+3}{a_0}]$, it satisfies $A_i \geq a_0$, and there exists no $i' \in \mathbb{N} \setminus \{ i \}$ such that $T_{i'} \in [\frac{7(\ell-1)}{a_0},\frac{7\ell}{a_0}]$, and let us call $\ell$ \emph{bad} otherwise. Note that if $\ell$ is good, then the corresponding $i$ is favourable. For fixed $\ell$, $\ell$ being good is equivalent to a Poisson process of intensity $\lambda$ having exactly one arrival in a time interval of length $1/a_0$ and no arrival in two disjoint intervals of length $2/a_0$ resp.\ $4/a_0$ that are also disjoint from the first interval. Hence, any fixed $\ell$ is good with probability $\frac{\lambda q}{a_0} e^{-7\lambda /a_0}$, and the events $\{ \ell \text{ is good} \}$, $\ell \in \N$, are independent.
Therefore, if we put
\[ S: = \inf \{ \ell \in \mathbb{N} \colon \ell \text{ is good} \}, \]
then $S$ is a geometric random variable with parameter (success probability) $\frac{\lambda q}{a_0} e^{-7\lambda /a_0}$. In particular, since $E_1 \leq \frac{7S}{a_0}$ almost surely, it follows that $\E(e^{\alpha E_1})<\infty$ for all $\alpha>0$ sufficiently small, as we claimed in Section~\ref{sec-k}.

In order to conclude that $\E(N_{E_1}^2)<\infty$, it now suffices to prove that $\E(N_{ \frac{7S}{a_0}}^2)<\infty$. Conditional on $S$, in the time interval $[\frac{7(S-1)}{a_0}, \frac{7S}{a_0}]$ precisely one trajectory is born, and thus $N_{ \frac{7S}{a_0}}$ is equal to one plus the sum of $S-1$ independent random variables $X_1,\ldots,X_{S-1}$, with $X_\ell$ describing the number of trajectories born in the respective bad time interval $[\frac{7(\ell-1)}{a_0}, \frac{7\ell}{a_0}]$. Here, using inclusion--exclusion for the reasons of badness of the individual time intervals, $X_\ell$ is at most one plus the sum of three independent random variables:
\begin{enumerate}[(i)]
\item a Poisson$(2\lambda/a_0)$-distributed one conditioned to be positive (corresponding to the number of trajectories born in $[\frac{7(\ell-1)}{a_0}, \frac{7(\ell-1)+2}{a_0}]$),
\item the sum of a Poisson$(\lambda q/a_0)$-distributed one conditioned to be unequal to $1$ (describing the number of trajectories of initial slope at least $a_0$ born in $[\frac{7(\ell-1)+2}{a_0}, \frac{7(\ell-1)+3}{a_0}]$) and an independent Poisson$(\lambda (1-q)/a_0)$ -distributed one conditioned to be positive (describing the number of trajectories of initial slope less than $a_0$ born in the same time interval),
\item and a Poisson$(4\lambda/a_0)$-distributed one also conditioned to be positive (and equal to the number of trajectories born in $[\frac{7(\ell-1)+3}{a_0}, \frac{7\ell}{a_0}]$),
\end{enumerate}
where we note that the independence of the number of trajectories with slope less than $a_0$ and the ones with slope at least $a_0$ born within a given time interval follows from the Colouring Theorem~\cite[p.~53]{K93}.
The proof of the fact that the sum $\sum_{\ell=1}^{S-1} X_\ell +1$ has a finite second moment follows from elementary properties of the Poisson distribution and can be carried out analogously to~\cite[Section 5.5, part 1.]{HGSTW24}, therefore we omit the details. The lemma follows.

\subsection*{Acknowledgements}
The authors thank Felix Hermann for interesting discussions and comments, and for providing the R code using which Figure~\ref{fig-LB} was created. He gave the first construction proving Proposition~\ref{prop-cousinen}, which was somewhat different from the one appearing in our proof. In the context of Remark~\ref{rem-k}, he was the first to observe that a slope of $-(2k-1)$ occurs in the case $k=2$. The authors also thank Mátyás Zelei for creating Figure~\ref{fig-contlines}, Anton Wakolbinger for historical remarks and comments, and two anonymous reviewers for insightful suggestions and comments, in particular regarding the presentation of the continued lines representation.

\end{document}